\documentclass[a4paper]{article}
\usepackage[english]{babel}
\usepackage{amscd}
\usepackage{amsmath}
\usepackage{amsthm}
\usepackage{amssymb}
\usepackage{graphicx}
\usepackage[left=2.5cm,right=2.5cm,
    top=3cm,bottom=3.5cm,bindingoffset=0cm]{geometry}

\usepackage[numbers, sort&compress]{natbib}

\theoremstyle{plain}
    \newtheorem{lemma}{Lemma}[section]
    \newtheorem{theorem}{Theorem}
    \newtheorem{consequence}{Corollary}[section]
    \newtheorem{statement}{Proposition}[section]

\theoremstyle{remark}
    \newtheorem{remark}{Remark}[section]
    \newtheorem{example}{Example}[section]
  \theoremstyle{definition}
    \newtheorem{definition}{Definition}

\newcommand{\Ker}[1]{\mathrm{Ker} \, #1}

\newcommand{\sgrad}[1]{\mathrm{sgrad} \, #1}

\newcommand{\rank}[1]{\mathrm{rank} \, #1}

\newcommand{\corank}[1]{\mathrm{corank} \, #1}

\newcommand{\diff}[1]{\mathrm{d}  #1}
\newcommand{\diffFX}[2]{ \dfrac{\partial #1}{\partial #2} }

\newcommand{\diffFXY}[3]{ \frac{\partial^2 #1}{\partial #2 \partial #3} }

\newcommand{\K}{\mathbb{K}}

\newcommand{\R}{\mathbb{R}}
\newcommand{\Complex}{\mathbb{C}}

\newcommand{\Id}{\mathrm{E}}
\newcommand{\D}{\mathrm{D}}
\newcommand{\T}{\mathrm{T}}
\newcommand{\Cont}{\mathrm{C}}
\newcommand{\const}{\mathrm{const}}
\newcommand{\CP}{\overline{\mathbb{C}}}
\newcommand{\RP}{\overline{\mathbb{R}}}

\newcommand{\Imm}[1]{\mathrm{Im} \, #1}

\newcommand{\tr}[1]{\mathrm{tr} \, #1}

\newcommand{\eps}{\varepsilon}
\newcommand{\wave}{\widetilde}
\newcommand{\LieBracket}{ [\, , ] }

\newcommand{\g}{\mathfrak{g}}
\newcommand{\h}{\mathfrak{h}}
\newcommand{\lCal}{\mathfrak{l}}
\newcommand{\e}{\mathfrak{e}}
\newcommand{\so}{\mathfrak{so}}

\newcommand{\ELL}{\mathfrak{D}}
\newcommand{\hyp}{\mathfrak{D}^{h}}
\newcommand{\foc}{\mathfrak{D}^{\Complex}}
\newcommand{\sP}{\mathfrak{sp}}
\newcommand{\sL}{\mathfrak{sl}}
\newcommand{\ad}{\mathrm{ad}}
\newcommand{\zenter}{\mathrm{Z}}
\newcommand{\F}{\mathcal{F}}

\newcommand\goth{\mathfrak}
\usepackage{colortbl}

\title{Singularities of bi-Hamiltonian systems}
\author{
Alexey Bolsinov\footnote{Dept. of Math. Sciences, Loughborough University. E-mail: A.Bolsinov@lboro.ac.uk} \,
 and Anton Izosimov\footnote{Dept. of Mechanics and Mathematics, Moscow State University. E-mail: a.m.izosimov@gmail.com}}
\begin{document}
\maketitle
\abstract{We study the relationship between singularities of bi-Hamiltonian systems and algebraic properties  of compatible Poisson brackets.   As the main tool, we introduce the notion of linearization of a Poisson pencil. From the algebraic viewpoint,  a linearized Poisson pencil can be understood as a Lie algebra with a fixed 2-cocycle. In terms of such linearizations, we give a criterion for non-degeneracy of singular points of bi-Hamiltonian systems and describe their types.}


\section{Introduction}


\subsection{Statement of the problem}

Since the pioneering work by Franco Magri \cite{Magri}, followed by the fundamental  papers by
I.\,Gelfand and I.\,Dorfman \cite{GD},
F.\,Magri and C.\,Morosi \cite{MM}, and A.\,Reiman and M.\,Semenov-Tyan-Shanskii \cite{Reiman},
 it has been well known that integrability of many systems
in mathematical physics, geometry and mechanics is closely related to their bi-Hamiltonian nature.
Bi-Hamiltonian structures have been discovered for almost all classical systems and, at the same time, by using the bi-Hamiltonian techniques, many new interesting and non-trivial examples of integrable systems have been found.
Moreover, this approach, based on a very simple and elegant notion of compatible Poisson structures, proved to be very powerful in the theory of integrable systems not only for constructing new examples, but also for explicit integration, separation of variables and
description of analytical properties of solutions.

The goal of the present paper is to show that the bi-Hamiltonian approach might also be extremely effective for qualitative analysis of the dynamics and, in particular,  in the study of singularities of integrable systems, especially in the case of many degrees of freedom where using other methods often leads to serious computational problems.  The relationship between the singularities of an integrable bi-Hamiltonian system and the underlying bi-Hamiltonian structure was already observed in \cite{Bolsinov} and has been discussed in a systematic way in \cite{biham}. This work can be considered as  a natural continuation and completion of the programme started in these papers.

Each finite-dimensional integrable system gives rise to the structure of a singular Lagrangian fibration on the phase space whose fibers, by definition, are connected components of common level sets of the first integrals of the system. According to the classical Arnold-Liouville theorem \cite{Arnold}, regular compact fibers are invariant Lagrangian tori with quasi-periodic dynamics.  Although almost  all solutions  lie on these tori,  the singularities of Lagrangian fibrations
are very important at least for the following reasons:
\begin{itemize}
\item  The most interesting solutions (such as equilibrium points,  homoclinic and heteroclinic orbits,  stable periodic solutions, etc.) are located on singular fibers.

\item Many analytic effects (e.g., Hamiltonian monodromy \cite{Duistermaat} which can be understoond as an obstruction  to the existence of global action-angle variables) are determined by the topology of singular fibers.

\item The global dynamics of a system is directly related to the structure of the associated Lagrangian fibration which, in turn, is determined by its singularities.

\item The structure of singularities plays an  important role  in the problem of  topological obstructions to integrability.
\end{itemize}

By now, there is quite an accomplished theory that describes the topology of singular Lagrangian fibrations and classifies the main types of bifurcations and singularities for integrable Hamiltonian systems
\cite{FZ, intsys, Delzant, Eliasson, Lerman, Lerman2, MZ, AL1, AL2, Kharlamov, BolOsh}.  However, the description of Lagrangian fibrations for concrete examples of integrable systems still remains a rather non-trivial task, especially for multidimensional systems.

In the case of integrable systems, a usual scheme for the topological analysis of the dynamics can be explained as follows. To understand the dynamical properties of a given system, we need to  study the structure of the associated Lagrangian fibration and its singularities, which are completely determined by the integrals of the system. First of all,  it is required to describe the singular set of the fibration, i.e.,  those points where the first integrals become dependent.  Next, for each singular point, one needs to analyse its local structure.  If the integrals  are given explicitly, both problems can be solved by straightforward computation. But in reality,  this analysis usually involves solving systems of algebraic equations and can be very complicated even in the case of two degrees of freedom.  If the system is bi-Hamiltonian, then its integrals can be derived from the corresponding pencil of compatible Poisson brackets. So that, in fact,  this scheme starts from the bi-Hamiltonian structure and can be illustrated by the following diagram:
$$
\begin{array}{c}
\mbox{Poisson pencil}\\
\downarrow\\
\mbox{First integrals of the system}\\
\downarrow\\
\mbox{Lagrangian fibration}\\
\downarrow\\
\mbox{Singular set and local structure of singularities}\\
\downarrow\\
\mbox{Dynamical properties (e.g., stability)}
\end{array}
$$

Thus the information we are interested in is  already contained in the Poisson pencil associated with a given system.  Is it possible to extract this information
 directly from the properties of this pencil, without intermediate steps involving explicit description and analysis of first integrals?
The answer is positive and in the present paper we suggest a method which reduces the analysis of singularities of bi-Hamiltonian systems to the study of algebraic properties of the corresponding pencil of compatible Poisson brackets. Since in many examples the underlying bi-Hamiltonian structure has a natural algebraic interpretation, the technology developed in this paper allows one to reformulate rather non-trivial analytic and topological questions related to the dynamics of a given system  into pure algebraic language, which
often leads to quite simple and natural answers.

The paper is focused on non-degenerate singularities of bi-Hamiltonian systems. Non-degenerate singular points of integrable systems are, in some sense, generic and analogous to
Morse critical points of smooth functions. In particular, they are stable under small perturbations and are linearizable  in the sense that  the Lagrangian fibration near a non-degenerate singular point is
symplectomorphic to the one given by quadratic parts of the integrals (Eliasson \cite{Eliasson}). Topologically (and even symplectically) a non-degenerate singularity can be represented as the product of ``elementary blocks'' of three possible types: elliptic, hyperbolic and focus.  The complete local
invariant of such a singularity is the (Williamson) type of the point, a triple  $(k_e, k_h, k_f)$ of non-negative integers being the numbers of  elliptic, hyperbolic and focus components in this decomposition (see Section \ref{ndsp} for precise definitions).
The Williamson type contains not only the complete information about the Lagrangian fibration, but also determines the dynamics nearby this point. For example, Lyapunov stability of a non-degenerate equilibrium point is equivalent to the absence of hyperbolic and focus components, i.e.,  $k_h = k_f = 0$.  A similar statement holds true for singular periodic solutions.

Let us make one important remark about compatible Poisson brackets considered in the present paper.
In general, there are two essentially different types of Poisson pencils.
\begin{enumerate}
\item Symplectic pencils: brackets forming a pencil are non-degenerate.   The integrals in this case are the traces of powers of the recursion
operator \cite{GD, MM}.

\item Kronecker pencils: brackets forming a pencil are all degenerate (see Definition \ref{Kronpencil} below).
The integrals in this case are the Casimir functions of these brackets.
This situation was studied in \cite{Reiman, GZ3, Bolsinov} and can be viewed as a generalization of the argument shift
method \cite{MF, Manakov}.

\end{enumerate}

In this paper we discuss the second situation only. A similar treatment of  symplectic pencils remains an open and, in our opinion, very interesting and important problem.

Let us briefly discuss the content and structure of the paper. The problem we are dealing with can be formulated as follows.  Consider a bi-Hamiltonian system on a smooth manifold $M$ related to a pencil of  compatible {\it degenerate} Poisson structures $$\Pi=\{P_\lambda=P_0 + \lambda P_\infty\}, \quad \rank \Pi < \dim M$$ where $$\rank \Pi = \max_{\substack{\lambda, x}}\, \rank P_\lambda (x).$$ The family $\mathcal F$ of first commuting integrals of this system consists of  the  Casimir functions of all generic $P_\lambda$'s.  We assume that these integrals are sufficient for complete integrability so that they define the structure of a Lagrangian fibration on generic symplectic leaves.  Our aim is to study the singularities of this fibration.

  Assume that $x\in M$ belongs to a generic symplectic leave  $O(x)$ (w.r.t. a certain Poisson structure, say $P_0$) and  is singular,  i.e.,  the dimension of the subspace
$\mathrm{d}\mathcal F(x)\subset \T^*_x O(x)$ generated by the differentials of the first integrals is smaller than $\frac{1}{2} \dim O(x)$. What can we say about the local structure of this singularity? Is it non-degenerate? If yes, what is the Williamson type of this singularity?

As was explained above, the answers to these questions are ``hidden'' in the local properties of the pencil $\Pi$ at the point $x$.  First of all, we notice that $x$ is singular if and only if for some $\lambda_i \in \CP$ the rank of $P_{\lambda_i}(x)$ drops, i.e.,  $\rank P_{\lambda_i}(x)< \rank \Pi$.  Then  for each $\lambda_i$,  we introduce the notion of $\lambda_i$-linearization of the pencil $\Pi$ at the point $x$  (Definition \ref{linearizationDef}, Section \ref{linpen}) playing a crucial role in our construction.

The linearization of a Poisson pencil is still a pencil of compatible Poisson brackets but of  much simpler nature,  namely,  one of the generators of this pencil is a linear Poisson structure and the other is a constant Poisson structure.
Speaking more formally, the linearization of $\Pi$ is a pair $(\mathfrak g, A)$ where $\mathfrak g$ is a finite-dimensional Lie algebra and $A$ is a 2-cocycle on it, i.e.,
a skew-symmetric form on $\goth g$ satisfying the cocycle condition
$$
A(\xi, [\eta,\zeta]) + A(\eta, [\zeta,\xi]) + A(\zeta, [\xi,\eta])  = 0,  \quad \xi,\eta,\zeta\in\goth g.
$$
Equivalently, $A$ can be  understood as a constant Poisson structure on $\mathfrak g^*$ compatible with the standard Lie-Poisson structure.

 Such linear pencils are discussed in Sections \ref{linpenFirst}, \ref{linPenSing}.  The zero element $0\in\mathfrak g^*$ is a singular point for the Lagrangian fibration associated with $(\mathfrak g, A)$, and in this ``linear'' case we can ask the same questions about structure, non-degeneracy and type of the singularity.
 If $0\in\mathfrak g^*$ is non-degenerate, we call the linear pencil {\it non-degenerate}.

Our first result is a non-degeneracy condition (Theorem \ref{nd1}, Section \ref{regcriterion}) which asserts that the singularity of the Lagrangian fibration associated with $\Pi$ at the point $x$ is non-degenerate if and only if the pencil $\Pi$ is diagonalizable at $x$,  and all $\lambda_i$-linearizations are non-degenerate.  Moreover, the topological type of the singularity at $x$ is the sum of the topological types of the corresponding $\lambda_i$-linearizations (Theorems \ref{type1}, \ref{blt}, Section \ref{regcriterion}).

Thus, this theorem reduces the problem to the case of linear Poisson pencils.
To complete our analysis, we need a tool allowing us to verify the non-degeneracy condition for linear pencils. In Section \ref{cndlp}, we reformulate the non-degeneracy condition for $(\goth g, A)$ in purely algebraic terms.  Namely, $\Ker A$ must be a Cartan subalgebra of $\goth g$ with some special properties (Theorem \ref{ndlp1}).   This restriction on $(\goth g, A)$  turns out to be so rigid that  we have succeeded in obtaining a complete description of non-degenerate linear pencils. This is done in Theorem \ref{nlpr}  which basically states that ``modulo a commutative part'' each non-degenerate pencil is the direct sum of ``elementary blocks'' of 6 possible  types.  Three of them are defined on the semisimple Lie algebras $\goth g=\goth{so}(3)$,  $\goth{sl}(2)$ and $\goth{so}(3,\mathbb C)$ with $A$ being  $A(\xi,\eta)=\langle a, [\xi,\eta]\rangle$ for $a\in \goth g$ regular and semisimple. The three others  are related to the so-called diamond Lie algebra $\mathfrak D$  (see Section  \ref{cndlp} for details), also known as the Nappi-Witten algebra.

As an application of this general scheme in Section \ref{TodaSection} we study the singularities of the classical periodic Toda lattice and deduce the stability of all its motions.
The proofs of all main results are given in Section \ref{proofsSection}.

We hope that these methods will be quite efficient for the qualitative analysis of bi-Hamiltonian systems, especially for those with many degrees of freedom where direct methods  lead to serious technical  difficulties. Some illustrations and applications of our approach can be found in \cite{Izosimov, Basak}.  Moreover,  we believe that this approach could be extended to the infinite-dimensional case too. In particular, a formal application of this scheme to the periodic KdV equation leads to a description of elliptic stable solutions, the conclusion so obtained agrees with the results of \cite{McKean, Kapeller}.


\subsection{Integrable systems and their singularities}\label{ndsp}

Let $(M^{2n}, \omega)$ be a symplectic manifold. Then $\Cont^\infty(M^{2n})$ is a Lie algebra with respect to the Poisson bracket.

\begin{definition}
A commutative subalgebra $\mathcal F \subset \Cont^\infty(M^{2n})$ is called \textit{complete at a point} $x \in M^{2n}$ if  $\dim \diff \F(x) = n$, where
	 $\diff \F(x) = \{ \diff f(x), f \in \F\} \subset \T^*_x M$. \par 
	 A commutative subalgebra $\mathcal F \subset \Cont^\infty(M^{2n})$ is complete on $M^{2n}$, if it is complete on an everywhere dense subset.
\end{definition}

Let $\mathcal F \subset \Cont^\infty(M^{2n})$ be a complete commutative subalgebra. Consider an arbitrary $H \in \F$ and the corresponding Hamiltonian vector field
$$
	\sgrad H = \omega^{-1}\diff H.
$$
Then all functions in $\F$ are pairwise commuting integrals of $\sgrad H$, and $\sgrad H$
is a \textit{completely integrable Hamiltonian system}. So, an integrable system can be understood as a complete commutative subalgebra $\mathcal F \subset \Cont^\infty(M^{2n})$ with a distinguished element $H \in \mathcal F$. However, a particular choice of $H \in \mathcal F$ is not important to us. For
 this reason, we will not distinguish between integrable systems and complete commutative subalgebras.

\begin{remark}
Note that as a vector space, $\F$ may be infinite-dimensional.
\end{remark}

Consider an integrable system $\F$. Then the common level sets $\{\F = \const\}$ define a \textit{singular Lagrangian fibration} on $M^{2n}$ associated with $\F$.

\begin{definition}
	A point $x \in M^{2n}$ is called \textit{singular} for $\F$ if $\dim \diff \F(x) < n$. The number $\dim \diff \F(x)$ is called \textit{the rank} of a singular point $x$.
\end{definition}

A regular fiber of a singular Lagrangian fibration  is a fiber which does not contain singular points. By the Arnold-Liouville theorem, all compact regular fibers of a singular Lagrangian fibration are tori, and the dynamics on these tori is quasi-periodic.  However, the most interesting solutions of an integrable system, such as fixed points and stable periodic trajectories, belong to singular fibers. That is why it is important to study singularities of Lagrangian fibrations.
\par
As it usually happens in singularity theory, is is not realistic to describe all possible singularities, so one should start with studying the most generic of them. The most generic singularities of an integrable system are the \textit{non-degenerate} ones defined below. Details can be found in \cite{BolOsh, Eliasson, AL1}.
\par

Let $\F$ be an integrable system on $(M^{2n}, \omega)$, and  $x \in M^{2n}$ be a singular point of $\F$. Suppose that $f \in \F$, and $\diff f(x) = 0$. Then we can consider the linearization of the vector field $\sgrad f$ at the point $x$ as a linear operator $A_{f}: \T_x M \to \T_x M$.  Let
$$
A_{\F} = \{ A_{f} \}_{f \in \F, \diff f(x) = 0}.
$$
 Since the flow defined by $\sgrad f$ preserves the symplectic structure, $A_{\F} \subset \sP(\T_{x}M)$.  Moreover, $\F$ is commutative and, therefore,  $A_{\F}$ is a commutative subalgebra of $\sP(\T_{x}M, \omega)$. \par

Now consider the space $W = \{\sgrad f(x), f \in \F\} \subset \T_{x}M$. Since $\F$ is commutative, $W$ is isotropic and all operators belonging to $A_{\F}$ vanish on $W$.  Consider the skew-orthogonal complement to $W$ with respect to $\omega$, i.e. the subspace
$$W^\bot = \{ \xi \in \T_x M ~|~  \omega (\xi, W)=0\}.$$  Obviously, $W\subset W^\bot$, and $W^\bot$ is invariant under $A_{\F} $. Consequently, we can consider  elements of $A_{\F}$ as operators on $W^{\bot} / W$. Since $W$ is isotropic, the quotient $W^{\bot}/W$ carries a natural symplectic structure induced by $\omega$, and  $A_{\F} $ is a commutative subalgebra in $\sP(W^{\bot} / W, \omega)$.

\begin{definition}\label{nd}
	A singular point $x$ is called \textit{non-degenerate}, if the subalgebra $A_{\F}$ constructed above is a Cartan subalgebra in $\sP(W^{\bot} / W, \omega)$.
\end{definition}

 If $A$ is an element of a Cartan subalgebra $\h\subset \sP(2m, \R)$, then its eigenvalues have the form
\begin{align*}
&\pm \lambda_{1}i, \dots, \pm \lambda_{k_{e}}i,\\
 &\pm \mu_{1}, \dots, \pm \mu_{k_{h}},\\
  &\pm \alpha_{1} \pm \beta_{1}i, \dots, \pm \alpha_{k_{f}} \pm \beta_{k_{f}}i,
\end{align*}
where $k_e + k_h + 2k_f = m$. The triple $(k_{e}, k_{h}, k_{f})$ is the same for almost all $A \in \h$. This triple is called the \textit{type} of the Cartan subalgebra $\h$. All Cartan subalgebras of the same type are conjugate to each other (Williamson, \cite{Williamson}).

\begin{definition}
	The \textit{type} of a non-degenerate singular point $x$ is the type of the  associated Cartan subalgebra $A_{\F} \subset \sP(W^{\bot} / W, \omega)$.
\end{definition}

For every non-degenerate singular point $x$ of rank $r$, the following equality holds:
$$
k_{e} + k_{h} + 2k_{f} = n - r.
$$\par

 Let us formulate the Eliasson theorem on the linearization of a Lagrangian fibration in the neighbourhood of a non-degenerate singular point. Define the following standard singularities.
 \begin{enumerate}
 \item The fibration given by the function $p^2 + q^2$ in the neighbourhood of the origin in $(\R^2, \diff p \wedge \diff q)$ is called an {\it elliptic} singularity.
  \item The fibration given by the function $pq$ in the neighbourhood of the origin in  $(\R^2, \diff p \wedge \diff q)$ is called a {\it hyperbolic} singularity.
       \item The fibration given by the commuting functions  $p_1q_1 + p_2q_2, p_1q_2 - q_1p_2$ in the neighbourhood of the origin in  $(\R^4, \diff p_1 \wedge \diff q_1 + \diff p_2 \wedge \diff q_2)$ is called a {\it focus-focus} singularity.
 \end{enumerate}

\begin{theorem}[Eliasson, see \cite{Eliasson}]
\label{EliassonThm}
	Let $\F$ be an integrable system  and $x$ be its singular point of rank $r$ and type $(k_e, k_h, k_f)$. Then the associated Lagrangian fibration is locally fiberwise symplectomorphic to the direct product of $k_e$ elliptic, $k_h$ hyperbolic, and $k_f$ focus-focus singularities, multiplied by a trivial non-singular fibration $\R^r \times \R^r$.
\end{theorem}

Thus, a non-degenerate singularity is completely determined by its type and rank.


\subsection{Bi-Hamiltonian systems and construction of the family $\F$}\label{ihs}

\begin{definition}
		Two Poisson brackets $P_0, P_\infty$ (on a smooth manifold $M$) are called \textit{compatible} if any linear combination of them is a Poisson bracket again. The set of non-zero linear combinations of compatible Poisson brackets $\Pi = \{\alpha P_0 + \beta P_\infty\}$ is called a \textit{Poisson pencil}.
	\end{definition}

	\begin{remark}\label{complCotSpace}
		Sometimes it is necessary to consider complex values of $\alpha$ and $\beta$. In this case $\alpha P_0 + \beta P_\infty$ should be treated as a complex-valued Poisson bracket on complex-valued functions. The corresponding Poisson tensor is a bilinear form on the complexified cotangent space at each point.
	\end{remark}

	Since it only makes sense to consider Poisson brackets up to proportionality, we write Poisson pencils in the form
	$$
		\Pi = \{ P_{\lambda} = P_{0} + \lambda P_{\infty}\}_{\lambda \in \CP}.
	$$

	\begin{definition}
		The \textit{rank} of a pencil $\Pi$ at a point $x \in M$ is defined to be
		$$
			\rank \Pi(x) = \max\limits_{\lambda} \rank P_{\lambda}(x).
		$$
			The rank of $\Pi$ on $M$ is
		$$
			\rank \Pi = \max\limits_{x} \rank \Pi(x) =  \max\limits_{x, \lambda} \rank P_{\lambda}(x).
		$$
	\end{definition}

	\begin{definition}
		A vector field $v$ is called \textit{bi-Hamiltonian} with respect to a pencil $\Pi$, if it is Hamiltonian with respect to all brackets of the pencil.
	\end{definition}

	Let $\Pi$ be a Poisson pencil, and let $v$ be a vector field which is bi-Hamiltonian with respect to $\Pi$. The following construction allows us to obtain a large family of integrals for $v$, which are in involution with respect to all brackets of $\Pi$ (see \cite{Reiman}).

	\begin{statement}\label{Magri}
	Let $\Pi = \{ P_{\lambda} \}$ be a Poisson pencil. Then\par
	\begin{enumerate}
		\item If $f$ is a Casimir function of $P_{\lambda}$ for some $\lambda$, then $f$ is an integral of any vector field that is bi-Hamiltonian with respect to $\Pi$.
		\item If $f$  is a Casimir function of $P_{\lambda}$, $g$  is a Casimir function of $P_{\mu}$ and $\lambda \neq \mu$, then $f$  and $g$ are in involution with respect to all brackets of the pencil.
		\item If $f$ and $g$ are Casimir functions of $P_{\lambda}$, and $\rank P_{\lambda}(x) = \rank \Pi$ for almost all $x \in M$, then $f$  and $g$ are in involution with respect to all brackets of the pencil.
	\end{enumerate}
	\end{statement}

	Let $\widehat{\mathcal{F}}$ be the system generated by all Casimir functions of all brackets of the pencil satisfying the condition $\rank P_{\lambda} = \rank \Pi$. Proposition \ref{Magri} implies that $\widehat{\mathcal{F}}$ is a family of integrals of $v$ in involution. Our goal is to study the singularities of $\widehat{\mathcal{F}}$. However, we need to slightly modify the definition of  $\widehat{\mathcal{F}}$ for the following two reasons.
	\begin{enumerate}
\item In general, we cannot guarantee that brackets of the pencil possess globally defined Casimir functions.
\item Even if globally defined Casimir functions do exist, their behavior may be unpredictable near those points where the rank of the corresponding bracket drops.
\end{enumerate}

Since the properties of the singularities we are going to discuss are local,  we will work in a small neighborhood of a singular point $x\in M$ and will use local Casimir functions of only those Poisson brackets  whose rank in this neighborhood does not drop.  More precisely,
the algebra of commuting integrals $\mathcal{F}$ considered throughout the paper is defined as follows.

Let $x \in M$ be such that $\rank \Pi(x) = \rank \Pi$.  Then we can find $\alpha\in\mathbb R$ such that $\rank P_{\alpha}(x) = \rank \Pi$.
	 Moreover, we can find a neighbourhood $U(x)$ and $\eps > 0$ such that for any $\nu$ satisfying $|\nu - \alpha| < \eps$ the following is true:

	 \begin{enumerate}
	 	\item $\rank P_{\nu}(y) = \rank \Pi$ for any $y \in U(x)$;
	 	\item $P_{\nu}$ admits $k$ independent local Casimir functions defined in $U(x)$, where $k=\corank \Pi$.
	 \end{enumerate}
Consider the family $\F = \F_{\alpha, \eps}$ generated, as a vector subspace of $\Cont^\infty(U(x))$, by all these Casimir functions. Proposition \ref{Magri} implies the following.

	\begin{statement}
		$\F$ is a (local) family of integrals in involution for any system which is bi-Hamiltonian with respect to $\Pi$.
	\end{statement}

	\begin{remark}\label{badRemark}
		The choice of $\alpha$ and $\eps$ is not important, which means that our results remain true for any choice of $\alpha, \eps$. Moreover, under some additional conditions, we will get the same family of integrals for all $\alpha, \eps$.
		What is important is that $\F$ is generated by the Casimir functions of brackets which are regular at the point $x$ (see Example \ref{badExample}). In applications,   $\F$ can be replaced by  the family $\widehat{\mathcal{F}}$ of global Casimirs under two additional assumptions, which are easy to verify:
\begin{enumerate}
\item Existence of global Casimirs: for almost all $P_\nu\in\Pi$ there are $k=\corank \Pi$ Casimir functions whose differentials are independent at $x$.

\item Smooth dependence of Casimirs on the parameter of the pencil: if $\rank P_{\lambda_0} = \rank \Pi$, then every Casimir $f(x)$ of $P_\lambda$ can be included into a family $f_\lambda(x)$, $\lambda \in (\lambda_0-\delta, \lambda_0+\delta)$ smoothly depending on $\lambda$ such that $f_\lambda(x)$ is a Casimir for $P_\lambda$ and $f_{\lambda_0}(x)=f(x)$.
\end{enumerate}		
		\end{remark}

	First of all, we need to discuss  completeness conditions for $\F$.

	\begin{definition} The \textit{spectrum} of a pencil $\Pi$ at a point $x$ is the set
	$$
		\Lambda(x) = \{ \lambda \in \CP: \rank P_{\lambda}(x) < \rank \Pi(x) \}.
	$$
	\end{definition}

	Let
	$$
		S = \{ x: \Lambda(x) \neq \emptyset\} .
	$$

	\begin{definition}
	\label{Kronpencil}
		A pencil $\Pi$ is \textit{Kronecker}, if the set $S$ has measure zero (i.e. if the spectrum is empty almost everywhere).
	\end{definition}

	In other words, a pencil is Kronecker if its Jordan-Kronecker decomposition has only Kronecker blocks (i.e. has no Jordan blocks) almost everywhere on $M$ (see Appendix A).

	\begin{theorem}[\cite{Bolsinov}]\label{BolTh}
	Assume that $\rank P_\alpha(x) = \rank \Pi$ and let $O(\alpha, x)$ be the symplectic leaf of $P_{\alpha}$ passing through $x$. Then $\F \mid_{O(\alpha, x)}$ is complete at $x$ if and only if $x \notin S$.
	\end{theorem}

	\begin{consequence}\label{intCor}
		$\F$ is complete on $O(\alpha, x)$ if and only if  the set $S \cap O(\alpha, x)$ has measure zero.		 
	\end{consequence}

	\begin{consequence}
		If $\Pi$ is Kronecker, then $\F$ is complete on almost all regular symplectic leaves.		
	\end{consequence}

	The theorem also implies that the singular points of $\F \mid_{O(\alpha, x)}$ are exactly the points where the rank of some bracket $P_{\beta}\in\Pi$ drops.
		The main goal of the present paper is to answer the following question: how to check non-degeneracy of these points and determine their type?
	 It turns out that the answer can be given in terms of the so-called linearization of the pencil $\Pi$, which is defined in the next section.

	
\section{Main constructions and results}
	

\subsection{Linearization of a Poisson pencil}\label{linpen}

Let $P$ be a Poisson structure on a manifold $M$ and $x \in M$. It is well-known that the linear part of $P$ defines a Lie algebra structure on the kernel of $P$ at  $x$. This structure is defined as follows. Let $\xi, \eta \in \Ker P(x)\subset \T^*_xM$. Choose any functions $f,g$ such that $\diff f = \xi, \diff g = \eta$, and set
$$
	[\xi, \eta] = \diff \{ f, g \}.
$$
\newpage
The following is well-known.
\begin{statement}
		If $\rank P(x) = \rank P$, then $\Ker P(x)$ is an Abelian Lie algebra.
\end{statement}

Consider a Poisson pencil $\Pi=\{ P_{\lambda} = P_{0} + \lambda P_{\infty} \}$ and fix a point $x$ such that $\rank \Pi(x) = \rank \Pi$. Denote by $\g_{\lambda}$ the Lie algebra on the kernel of $P_{\lambda}$ at the point $x$. For regular $\lambda$ (i.e. for $\lambda \notin \Lambda(x)$) the algebra $\g_{\lambda}$ is Abelian. For singular $\lambda$ (i.e., $\lambda \in \Lambda(x)$) this is not necessarily the case, and therefore $\g_{\lambda}$ carries non-trivial information about the behavior of the pencil in the neighborhood of $x$.

\begin{remark}
	For $\lambda \in \RP$, the algebra $\g_\lambda$ is real. However, for complex values of $\lambda$, the space $\Ker P_\lambda(x)$ is a subspace of $\T^*_xM \otimes \Complex$,  and therefore $\g_\lambda$ is considered as a complex Lie algebra (see Remark \ref{complCotSpace}).
\end{remark}

It turns out that, apart from the Lie algebra structure, $\g_{\lambda}$ carries one more additional structure.

\begin{statement}
	For any $\alpha$ and $\beta$  the restrictions of $P_{\alpha}(x)$ and $P_{\beta}(x)$ to $\g_{\lambda}(x)$ coincide up to a constant factor.
\end{statement}

\begin{proof}[Proof]
	Since $P_{\lambda}$ vanishes on $\g_{\lambda}$, all other brackets of the pencil are proportional.
\end{proof}

The restriction $P_{\alpha}|_{\g_{\lambda}}$ is a $2$-form on $\g_{\lambda}\subset \T^*_xM$ and, therefore,  can be interpreted as a constant Poisson tensor on $\g_{\lambda}^{*}$.

\begin{statement}\label{agree}
	The 2-form $P_{\alpha}|_{\g_{\lambda}}$, interpreted as a constant Poisson tensor on $\g_{\lambda}^{*}$, is compatible with the Lie-Poisson bracket on $\g_{\lambda}^{*}$.
\end{statement}

The proof follows from the next well-known statement.

\begin{statement}\label{linAgree}
	Let $\g$ be a Lie algebra, and let $A$ be a skew-symmetric bilinear form on $\g$. Then $A$ is compatible with the Lie-Poisson bracket if and only if it is a $2$-cocycle, i.e.
	\begin{align*}
		\diff A(\xi,\eta,\zeta) = A([\xi,\eta],\zeta) + A([\eta,\zeta],\xi) + A([\zeta,\xi],\eta) = 0
	\end{align*}
	for any $\xi, \eta, \zeta \in \g$.
\end{statement}

\begin{proof}[Proof of Proposition \ref{agree}]
	Since $P_{\alpha}$ and $P_{\lambda}$ are compatible, we have
	\begin{align*}
		\{ \{f,g\}_{\alpha}, h\}_{\lambda} &+ \{ \{g,h\}_{\alpha}, f\}_{\lambda} + \{ \{h,f\}_{\alpha}, g\}_{\lambda} + \\ &+ \{ \{f,g\}_{\lambda}, h\}_{\alpha} + \{ \{g,h\}_{\lambda}, f\}_{\alpha} + \{ \{h,f\}_{\lambda}, g\}_{\alpha} = 0.
	\end{align*}
	If $\diff f, \diff g, \diff h \in \Ker P_{\lambda}$, then the first three terms vanish, and
	\begin{align*}
		 \{ \{f,g\}_{\lambda}, h\}_{\alpha} + \{ \{g,h\}_{\lambda}, f\}_{\alpha} + \{ \{h,f\}_{\lambda}, g\}_{\alpha} = 0.
	\end{align*}
	So, for any $\xi, \eta, \zeta \in \g_\lambda$ we have
	$$
		P_{\alpha}([\xi, \eta], \zeta) + P_{\alpha}([\eta, \zeta], \xi) + P_{\alpha}([\zeta,\xi], \eta) = 0,
	$$	
	q.e.d.
\end{proof}

Consequently, $P_{\alpha}|_{\g_{\lambda}}$ defines a Poisson pencil on $\g_{\lambda}^{*}$. Denote this pencil by $\diff_{\lambda} \Pi(x)$.

\begin{definition}\label{linearizationDef}
	The pencil $\diff_{\lambda} \Pi(x)$ is called the \textit{$\lambda$-linearization} of $\Pi$ at the point $x$.
\end{definition}

The pencil $\diff_{\lambda} \Pi(x)$ is generated by a linear and a constant Poisson bracket. We call such pencils \textit{linear} and discuss some of their basic properties in the next section.


\subsection{Linear pencils}\label{linpenFirst}

\begin{definition}
	Let $\g$ be a Lie algebra, and  $A$ be a skew-symmetric bilinear form on $\g$.
	Then $A$ can be considered as a Poisson tensor on the dual space $\g^{*}$. Assume that the corresponding bracket is compatible with the Lie-Poisson bracket. The Poisson pencil $\Pi^{\g, A} = \{ P^{\g, A}_{\lambda} \}$, where
	$$
		P^{\g, A}_{\lambda}(x)(\xi, \eta) = \langle x, [\xi, \eta] \rangle + \lambda A(\xi, \eta),\mbox{ for } \xi, \eta \in \g\simeq \T^*_x\g^*,$$
		is called the \textit{linear pencil} associated with the pair $(\g, A)$.
\end{definition}

By Proposition \ref{linAgree}, a linear pencil is a Lie algebra $\g$ equipped with a $2$-cocycle $A$.

\begin{example}[Frozen argument bracket]
	Let $\g$ be an arbitrary Lie algebra, and $a \in \g^{*}$. Then the form $A_{a}(\xi, \eta) = \langle a, [\xi, \eta] \rangle$ is a $2$-cocycle (moreover, it is a coboundary). The corresponding Poisson bracket is called the \textit{frozen argument bracket}.
	This bracket naturally appears in the context of the argument shift method \cite{MF, Manakov}. For this reason, linear pencils associated to $A_{a}$ are called \textit{argument shift pencils}.
\end{example}

Following the general scheme from Section \ref{ihs}, we want to use linear pencils to construct commuting functions. To that end, we need the following property of \textit{regularity}.

\begin{definition}\label{regcocycle}
	We say that a cocycle $A$ on $\g$ is \textit{regular}, if $\rank \Pi^{\g, A} = \rank A$.
\end{definition}

\begin{example}
	For $A = A_{a}$, regularity of $A$ is equivalent to regularity of the element $a\in\g^*$.
\end{example}

Suppose that $A$ is a regular $2$-cocycle on $\g$. Then we can apply the construction of Section \ref{ihs} to the pencil $\Pi^{\g, A}$ and obtain the commutative family $\F$ defined in the neighbourhood of the origin $0\in\g^*$.

\begin{definition}
	Let $A$ be a regular $2$-cocycle on $\g$. The pencil  $\Pi^{\g, A}$ is called \textit{integrable} if $\F$ is complete on the symplectic leaf of $A$ passing through the origin.
\end{definition}

Thus, if a pencil  $\Pi^{\g, A}$ is integrable, then it canonically defines an integrable system on the symplectic leaf of $A$ passing through the origin.  This integrable system and its singularity at the origin are discussed in the next section.


\subsection{Singularities associated with integrable linear pencils}\label{linPenSing}

Consider an integrable linear pencil $\Pi^{\g,  A}$  and  the integrable system associated to  it  on the symplectic leaf of $A$ passing through the origin. The origin is a zero-rank singular point for this system. This means that every integrable linear pencil canonically defines a zero-rank singularity, i.e. a germ of an integrable system at a zero-rank singular point. Denote the singularity associated with $\Pi^{\g,  A}$ by $\mathsf{Sing}(\Pi^{\g,  A})$. Our goal is to understand under which conditions on $\g$ and $A$ this singularity is non-degenerate and, if so, to determine its type. We start with simple examples.

\begin{example}Let $\g$ be a real semisimple Lie algebra with two-dimensional coadjoint orbits and $A = A_{a}$ be an argument shift form where $a \in \g\simeq \g^{*}$ is a regular element. Below are the corresponding singularities.
\begin{itemize}
	\item $\so(3) \to$ an elliptic singularity. See Figure \ref{ellpic}.
	\begin{figure}
	\centerline{\includegraphics[scale=0.5]{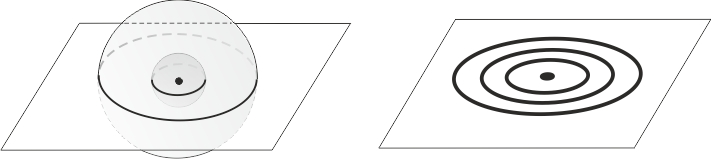}}
	\caption{Singularity corresponding to $\so(3)$}\label{ellpic}
	\end{figure}
	\item $\sL(2) \to $ a hyperbolic singularity if the Killing form is positive on $a$ (Figure \ref{hyppic}), an elliptic singularity if it is negative (Figure \ref{hyp1pic}), and degenerate if it is zero (Figure \ref{hyp2pic}).\par
		\begin{figure}
		\centerline{\includegraphics[scale=0.4]{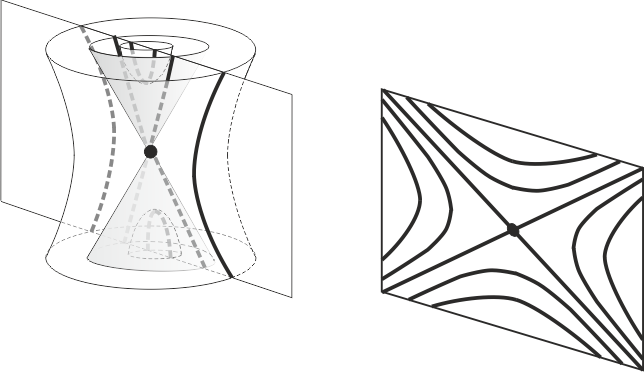}}
			\caption{Singularity corresponding to $\sL(2)$ with $\langle a,a \rangle > 0$}\label{hyppic}
	\end{figure}
			\begin{figure}
		\centerline{\includegraphics[scale=0.4]{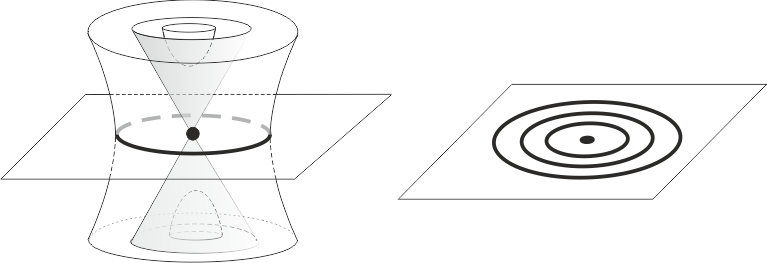}}
			\caption{Singularity corresponding to $\sL(2)$ with $\langle a,a \rangle < 0$}\label{hyp1pic}
	\end{figure}
		\begin{figure}
		\centerline{\includegraphics[scale=0.4]{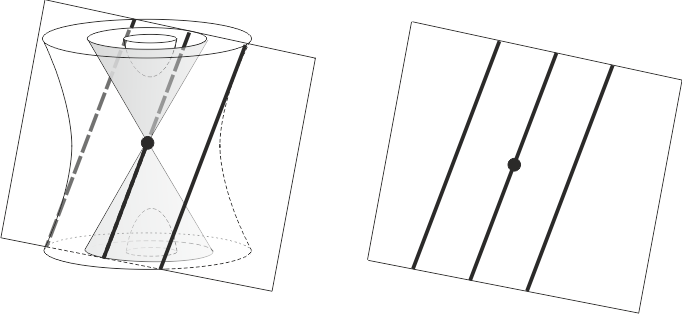}}
			\caption{Singularity corresponding to $\sL(2)$ with $\langle a,a \rangle = 0$}\label{hyp2pic}
	\end{figure}
\end{itemize}
\begin{example}
Let $\g$ be a real semisimple Lie algebra with four-dimensional coadjoint orbits and $A = A_{a}$ be an argument shift form where $a \in \g\simeq \g^{*}$ is a regular element. Below are the corresponding singularities.
\begin{itemize}
	\item  $\so(4) \simeq \so(3) \oplus \so(3) \to $ a center-center singularity, i.e., the  product of two elliptic singularities.
	\item  $\so(2,2) \simeq \sL(2) \oplus \sL(2) \to$  saddle-saddle (the product of two hyperbolic singularities), saddle-center (a product of an elliptic and a hyperbolic singularity), center-center (the product of two elliptic singularities), or a degenerate singularity.	
	\item  $\so(3,1) \simeq \so(3, \Complex) \simeq \sL(2, \Complex) \to$ a focus-focus singularity if $a$ is semisimple, and a degenerate singularity otherwise.
\end{itemize}
\end{example}

	Further, we show that no semisimple Lie algebras except for the sums of $\so(3), \sL(2)$ and $\so(3, 1)$ give rise to non-degenerate singularities. The counterpart of this fact in the theory of integrable systems is the Eliasson theorem: all non-degenerate singularities are products of elliptic, hyperbolic and focus-focus singularities (see Theorem \ref{EliassonThm}). However, some non-semisimple Lie algebras, as we shall see below, may ``produce'' non-degenerate singularities too.
\end{example}

\begin{definition}\label{nondegpencil}
	An integrable linear pencil $\Pi^{\g,  A}$ is called \textit{non-degenerate}, if the singularity $\mathsf{Sing}(\Pi^{\g,  A})$  is non-degenerate.
\end{definition}

 In Section \ref{cndlp}, we reformulate this definition in algebraic terms and classify all non-degenerate linear pencils.


\subsection{Main theorems}\label{regcriterion}

Let $\Pi=\{P_\lambda\}$ be a Poisson pencil of Kronecker type on $M$  (Definition \ref{Kronpencil}), and $x\in M$ be such that $\rank \Pi(x)=\rank \Pi$. In a neighborhood of $x$, we consider the commutative family $\mathcal F$ defined in Section \ref{ihs}. To state necessary and sufficient conditions for $x$ to be a non-degenerate singular point of $\mathcal F$, we need to introduce one algebraic condition on the pencil $\Pi(x)$.

\begin{definition}\label{diagpencil}
	A pencil $\Pi$ will be called diagonalizable at a point $x$, if for each $\lambda \in \Lambda(x)$ and any $\alpha \neq \lambda$ the following is true
	$$
		\dim \Ker \left(P_{\alpha}(x)\mid_{\Ker P_{\lambda}(x)}\right) = \corank \Pi(x).
	$$
\end{definition}

\begin{remark}
If $P_0$ and $P_\infty$ were non-degenerate, then the spectrum of the pencil would be just the spectrum of the recursion operator $\mathcal{R} = -P_\infty^{-1}P_0$, and the diagonalizability condition would mean the diagonalizability of $\mathcal{R}$. Since $P_0$ and $P_\infty$ are degenerate, it is not possible to define the operator $\mathcal{R}$ on the whole cotangent space. However, it can be defined on a certain quotient space (see Section \ref{zeroOrderTheory}). The diagonalizability condition means that this operator on the quotient is diagonalizable.\par
	In terms of the Jordan-Kronecker decomposition of the pencil $\Pi$ at $x$ the diagonalizability condition means that all the Jordan blocks $J(\lambda_{i})$ have size $1 \times 1$, i.e., are trivial (see Appendix A).
\end{remark}

\begin{theorem}[Non-degeneracy condition]\label{nd1}
Let  $x\in M$ and $\alpha \in \RP$ be such that $\rank P_\alpha(x)=\rank \Pi$. Let also $O(\alpha, x)$ be the symplectic leaf of the bracket $P_{\alpha}$ passing through $x$. 
	Then  $x$  is a non-degenerate singular point of the integrable system $\F\mid_{O(\alpha, x)}$  if and only if the following two conditions hold.
	\begin{enumerate}
		\item $\Pi$ is diagonalizable at $x$.
		\item For each $\lambda \in \Lambda(x)$ the $\lambda$-linearization $\diff_{\lambda}\Pi(x)$ is non-degenerate.
	\end{enumerate}
\end{theorem}
The proof of this theorem is given in Section \ref{mainProof}.
\newpage
\begin{remark}
Note that under conditions 1-2 of Theorem \ref{nd1} the pencil $\Pi$ is Kronecker, and the family $\F\mid_{O(\alpha, x)}$ is a completely integrable system.
\end{remark}
\smallskip
Recall that $\F$, by definition, is generated by Casimir functions of regular brackets. The following example shows that the statement of Theorem \ref{nd1} may be wrong if  $\mathcal F$ contains a Casimir function of a bracket which is singular at $x$.

\begin{example}\label{badExample}
	Consider $\so(3)^{*}$ with the following Poisson structure
	$$
		P_{0} = (x^{2}+y^{2}+z^{2})P_{\so(3)},
	$$
	where $P_{\so(3)}$ is the standard Lie-Poisson structure on $\so(3)^*$, and
	let $P_{\infty}$ be any constant bracket of rank two. It is easy to check that $P_{0}$ and $P_{\infty}$ are compatible. \par
	Choose $x^{2} + y^{2} + z^{2}$ as a Casimir function of $P_0$. The restriction of this function to the symplectic leaf of $P_{\infty}$ passing through the origin defines an integrable system. The origin is a non-degenerate elliptic singular point of this system. However, the linearization of the pencil at the origin is zero, therefore the conditions of Theorem \ref{nd1} do not hold. \par
	However, if we take a Casimir function of a regular bracket, it will look like 
	$$
	(x^{2} + y^{2} + z^{2})^{2} + \mbox{linear terms,}
	$$ 
	and its restriction to the symplectic leaf of $P_{\infty}$ will be degenerate, as predicted by Theorem \ref{nd1}.\par
	The problem is that  the function $x^{2}+ y^{2}+z^{2}$ is an ``isolated'' point in the set
of {\it all} Casimir functions of {\it all} brackets of the pancil.
	If the set of all Casimir functions formed a smooth family, then Theorem \ref{nd1} could be applied even if $\mathcal F$ contained Casimir functions of a singular bracket. This can be easily shown by continuity argument.
	\end{example}

\begin{theorem}[Type theorem]\label{type1}
	Assume that conditions 1-2 of Theorem \ref{nd1} hold. Then the type of the singular point $x$  is the sum of types of $\mathsf{Sing}(\diff_{\lambda}\Pi(x))$ for all $\lambda \in \Lambda(x)$. In other words, the type of $x$ is $(k_{e}, k_{h}, k_{f})$, where
	\begin{align*}
		k_{e} &= \sum_{\substack{\lambda \in \Lambda(x) \cap \RP}}k_{e}(\lambda),\qquad
		k_{h} =  \sum_{\substack{\lambda \in \Lambda(x) \cap \RP}}k_{h}(\lambda), \qquad
		k_{f} =  \sum_{\substack{\lambda \in \Lambda(x),\\ \Imm \lambda \geq  0}}k_{f}(\lambda),		
	\end{align*}
	and $(k_{e}(\lambda), k_{h}(\lambda), k_{f}(\lambda))$ is the type of $\mathsf{Sing}(\diff_{\lambda}\Pi(x))$.\end{theorem}
\begin{remark}
 The sums for $k_e$ and $k_h$ are taken over the real part of the spectrum, because the singularity $\mathsf{Sing}(\diff_{\lambda}\Pi(x))$ is always focus-focus if $\lambda$ is not real.

\end{remark}
The proof of Theorem \ref{type1} is given in Section \ref{mainProof}.
\par
\smallskip
 Taking into account the Eliasson theorem (Theorem \ref{EliassonThm}),  we can reformulate Theorem \ref{type1} as follows.

\begin{theorem}[Bi-Hamiltonian linearization theorem]\label{blt}
	Assume that conditions 1-2 of Theorem \ref{nd1} hold. Then the singular Lagrangian fibration of the system $\F\mid_{O(\alpha, x)}$ is locally symplectomorphic to
	$$
		\left(\prod_{\substack{\lambda \in \Lambda(x),\\ \Imm \lambda \geq 0}} \mathsf{Sing}(\diff_{\lambda}\Pi(x))\right) \times (\R^{k} \times \R^{k}),
	$$
	where $\R^{k} \times \R^{k}$ is a trivial non-singular Lagrangian fibration, and $k$
is the rank of $x$.\par

	 In other words, the Lagrangian fibration of a bi-Hamiltonian system is locally symplectomorphic to the direct product of the Lagrangian fibrations of its $\lambda$-linearizations and a trivial fibration.
\end{theorem}


\subsection{Description of non-degenerate linear pencils}\label{cndlp}
Suppose that $\Pi^{\g,  A}$ is a linear pencil. Then $\Ker A$ is necessarily a subalgebra (by the cocycle identity). Moreover, if $A$ is regular in the sense of Definition \ref{regcocycle}, then $\Ker A$ is an Abelian subalgebra (see Section \ref{race}).

Suppose that all elements of $\Ker A$ are $\ad$-semisimple, i.e. $\Ker A$ is an $\ad$-diagonalizable subalgebra. Then $\g$ admits a  ``root'' decomposition
$$
  	\g \otimes \Complex =  \Ker A \otimes \Complex + \sum(V_{\lambda_{i}} + V_{-\lambda_{i}}),
$$
where $\pm \lambda_i \in (\Ker A)^* \otimes \Complex$ are roots, and $V_{\pm \lambda_{i}}$ are root spaces, which means that for any $\xi \in \Ker A \otimes \Complex$, $e_{\lambda_i} \in V_{\lambda_i}$ we have $[\xi, e_{\lambda_i}] = \lambda_i(\xi)e_{\lambda_i}$. Notice that in our situation, the roots come in pairs $\pm \lambda_i$, since the operator $\ad_\xi$ belongs to $\sP(\g / \Ker A, A)$  for any $\xi \in \Ker A$ by the cocycle identity.\par
Obviously, the maximal possible number of linearly independent roots is
$$
n = \frac{1}{2}(\dim \g - \dim \Ker A) = \frac{1}{2}\rank A.
$$

The following theorem gives algebraic conditions for a linear pencil to be non-degenerate in the sense of Definition \ref{nondegpencil}.

\begin{theorem}\label{ndlp1}
	A linear pencil $\Pi^{\g,  A}$ is non-degenerate if and only if $\Ker A \subset \g$ is an Abelian subalgebra, and there exists a root decomposition
\begin{align}\label{ndlp1decomp}
  	\g \otimes \Complex = \Ker A \otimes \Complex  + \sum\limits_{i=1}^{n}(V_{\lambda_{i}} + V_{-\lambda_{i}}),
  \end{align}
where \begin{enumerate}\item $\lambda_1, \dots, \lambda_n \in (\Ker A)^* \otimes \Complex$ are linearly independent. \item Each space $V_{\pm \lambda_i}$ is one-dimensional. \item For all $\xi \in \Ker A \otimes \Complex$, $e_{\lambda_i} \in V_{\lambda_i}$ we have $[\xi, e_{\lambda_i}] = \lambda_i(\xi)e_{\lambda_i}$. \end{enumerate}
\smallskip
	The type of $\mathsf{Sing}(\Pi^{\g, A})$ in this case is $(k_{e}, k_{h}, k_{f})$ where $k_{e}$ is the number of   purely imaginary roots, $k_{h}$ is the number of real roots and $k_{f}$ is the number of pairs of complex conjugate roots.
\end{theorem}
\begin{remark}
If $\g$ is a complex Lie algebra, then we do not need to take the tensor product with $\Complex$.
\end{remark}

This algebraic condition makes it possible to classify all non-degenerate linear pencils completely. To state this result, first we need to define three special Lie algebras. Denote by $\ELL$ the real Lie algebra with basis elements $e,f,h,t$ satisfying the following relations:
	\begin{align*}
		[e,f] = h,\quad [t,e] = f,\quad [t,f] = -e.
	\end{align*}
	This algebra is known as the Diamond Lie algebra \cite{Kirillov}, or the Nappi-Witten Lie algebra \cite{NW}. Denote the complexification of the Diamond Lie algebra by $\foc$. The algebra $\foc$ has another real form, which is generated by $e,f,h,t$ with the following relations:
	\begin{align*}
		[e,f] = h,\quad [t,e] = e,\quad [t,f] = -f.
	\end{align*}
	Denote it by $\hyp$.

\begin{remark}
	The algebras $\ELL$ and $\hyp$ are  the only non-trivial one-dimensional central extensions of $\e(2)$ and $\e(1,1)$ respectively.

\end{remark}

	Let $\g \simeq \ELL$ be the Diamond Lie algebra, and let $a \in \g^*$ be such that $$a(h) = 1,\quad a(e) = a(f) = a(t) = 0.$$ Consider the argument shift cocycle $A_{a}(\xi, \eta) = \langle a, [\xi, \eta] \rangle$. Let us show that the singularity corresponding to the pencil $\Pi^{\g,  A_a}$ is non-degenerate elliptic. The Casimir functions of the Lie-Poisson bracket are given by
	\begin{align*}
	f_1 = h, \quad f_2 = e^2 + f^2 + 2th.	
	\end{align*}
	The family $\F$ is generated by $f_1,f_2$ and $t$, which is a Casimir of $A_a$. The symplectic leaf of $A_a$ passing through the origin is given by $\{t = 0, h = 0\}$. The restriction of $\F$ to this leaf is generated by one single function, namely $e^2 + f^2$. Consequently, the corresponding singularity is indeed non-degenerate elliptic.\par
Analogously, $\hyp$ and $\foc$ correspond to a hyperbolic and a focus-focus singularities. Also, we have already seen that the algebras $\so(3), \sL(2)$ and $\so(3, 1)$ define non-degenerate singularities. It turns out that all Lie algebras admitting non-degenerate linear pencils can be obtained from these six algebras by means of three elementary operations:  direct sum,  quotient by a central ideal, and adding an Abelian Lie algebra.

\begin{theorem}[Classification of complex non-degenerate linear pencils]\label{nlpc}
	A complex linear pencil $\Pi^{\g,  A}$ is non-degenerate if and only if the following two conditions hold.
	\begin{enumerate}
	\item The algebra $\g$ can be represented as
		\begin{align}\label{ndadecomp_c}
		\g \simeq \bigoplus \so(3, \Complex) \oplus  \left(\bigoplus \foc \right) / \lCal_{0}  \oplus V,
	\end{align}	
	where $V$ is Abelian, and $\lCal_{0}\subset \bigoplus \foc$ is a central ideal.
	 \item $\Ker A$ is a Cartan subalgebra of $\g$.
	 \end{enumerate}	
\end{theorem}

The proof is given in Section \ref{nlpProofComplex}.
\newpage
  \begin{remark}
  	Recall that a Cartan subalgebra $\h$ of a Lie algebra $\g$ is a self-normalizing nilpotent subalgebra. For $\g$ belonging to the list (\ref{ndadecomp_c}), a Cartan subalgebra is the same as a maximal $\ad$-semisimple Abelian subalgebra.

	  \end{remark}

\begin{theorem}[Classification of real non-degenerate linear pencils]\label{nlpr}
	A real linear pencil $\Pi^{\g,  A}$ is non-degenerate if and only if the following two conditions hold.
	\begin{enumerate}
	\item The algebra $\g$ can be represented as
	\begin{align}\label{ndadecomp}
		\g \simeq \bigoplus \so(3) \oplus \bigoplus \sL(2) \oplus \bigoplus \so(3, \Complex) \oplus \left( \bigoplus \ELL  \oplus \bigoplus \hyp    \oplus \bigoplus \foc \right) / \lCal_{0}   \oplus V,
	\end{align}
	where $V$ is Abelian, and $\lCal_{0} \subset \bigoplus \ELL  \oplus \bigoplus \hyp    \oplus \bigoplus \foc $ is a central ideal.
	 \item $\Ker A$ is a Cartan subalgebra of $\g$.
	 \end{enumerate}	
	 If a real linear pencil $\Pi^{\g,  A}$ is non-degenerate, then the type of $\mathsf{Sing}(\Pi^{\g,  A})$ is $(k_{e}, k_{h}, k_{f})$, where
	\begin{itemize}
		\item $k_{e}=$ the number of $\so(3)$ terms in \eqref{ndadecomp} $+$ the number of $\ELL$ terms in \eqref{ndadecomp} $+$ the number of $\sL(2)$ terms in \eqref{ndadecomp} such that the Killing form on $\sL(2) \cap \Ker A$ is negative;
		\item $k_{h}=$ the number of $\hyp$ terms in \eqref{ndadecomp}  $+$  the number of $\sL(2)$ terms in \eqref{ndadecomp} such that the Killing form on $\sL(2) \cap \Ker A$ is positive;
		\item $k_{f}=$ the number of  $\so(3, \Complex)$ terms in \eqref{ndadecomp} $+$ the number of $\foc$ terms in \eqref{ndadecomp}.
	\end{itemize}	

\end{theorem}

The proof of the first part of the theorem is given in Section \ref{nlpProofReal}. The proof of the second part is given in Section \ref{typeProof}.


\section{Periodic Toda lattice}\label{TodaSection}
In this section we consider the classical periodic Toda lattice.
As it was shown by Foxman and Robbins \cite{Robbins}, singularities of this system coincide with those points where the corresponding Lax operator has double periodic or anti-periodic eigenvalues, and all singularities are non-degenerate elliptic. Their approach is based on the study of the higher Lax flows, i.e. the Lax flows corresponding to the integrals of the system. Our aim is to show that these results are easily obtainable by the bi-Hamiltonian approach, without considering the integrals or the corresponding Lax flows. As it will be seen later, the singularities of the Toda lattice are, in essence, defined by the geometry of the corresponding quadratic Poisson bracket.
\subsection{Bi-Hamiltonian structure}
The phase space $M_{T}$ of the periodic Toda lattice with $n$ sites is $\R_{+}^n \times \R^n$ endowed with Flaschka variables $$a = (a_1, \dots, a_n) \in \R_{+}^n, \quad b = (b_1, \dots, b_n) \in \R^n. $$
It is convenient to treat $a$ and $b$ as infinite $n$-periodic sequences
$$a \in \R_{+}^\infty, a_{i+n} = a_i, \quad b \in \R^\infty, b_{i+n}=b_i.$$
The equations of motion are
\begin{align*}
\begin{cases}
\dot a_i = a_i(b_{i+1}-b_i),\\
\dot b_i = 2(a_i^2 - a_{i-1}^2).
\end{cases}
\end{align*}
It is well known that these equations are bi-Hamiltonian. The corresponding pencil $\Pi_T$ is given by
\begin{align*}
\begin{aligned}
\{a_i, b_i\}_0 = a_ib_i, \quad  \{a_{i}, b_{i+1}\}_0 = -a_{i}b_{i+1}, &\quad \{a_i, a_{i+1}\}_0 = -\frac{1}{2}a_ia_{i+1},  \quad \{b_i, b_{i+1}\}_0 =- 2a_i^2,\\
\{a_i, b_i\}_\infty = a_i,  &\quad \{a_{i}, b_{i+1}\}_\infty = -a_{i}.
\end{aligned}
\end{align*}
The corresponding Hamiltonians are
\begin{align*}
H_\lambda = \begin{cases}
\sum\limits_{i=1}^n b_i  \quad \mbox{for } \lambda \neq \infty,\\
 \sum\limits_{i=1}^n a_i^2 + \dfrac{1}{2} \sum\limits_{i=1}^nb_i^2  \quad \mbox{for } \lambda = \infty.
\end{cases}
\end{align*}
\begin{remark}
Different authors use different forms of compatible Poisson structures for the Toda lattice. The above form is similar to the one used in \cite{GZToda, Toda2}.
\end{remark}
The integrals of the Toda lattice are the Casimir functions of the pencil $\Pi_T$. Taking all of these Casimirs,
we obtain an integrable system $\F_{T}$. Our goal is to study the singularities of this system. First, let us make several preliminary remarks.
\begin{enumerate}
\item The pencil $\Pi_T$ admits globally defined Casimirs depending smoothly on $\lambda$, which justifies the application of our scheme to the study of singularities of $\F_{T}$ (see Remark \ref{badRemark}).
\item Our approach is only suitable for points $x$ satisfying $\rank \Pi_T(x) = \rank \Pi_T$. However, since the rank of $P_\infty$ is constant (since $a_i>0$), all points $x \in M_{T}$ satisfy this condition.
\item Our results can be applied to the restriction of $\F_{T}$ to any symplectic leaf of $P_\infty$, since all symplectic leafs of this bracket are regular.
\end{enumerate}

{ According to the general scheme, in order to study the singularities of $\F_{T}$, we need to do the following.
\begin{enumerate}
	\item For each point $x \in M_{T}$ determine the spectrum of the pencil at $x$. The point $x$ is singular if and only if the spectrum is non-empty.
	\item If $x$ is singular, then for each $\lambda$ in the spectrum, check the following conditions:
	\begin{itemize}
		\item $\dim \Ker \left(P_{\infty}(x)\mid_{\Ker P_{\lambda}(x)}\right) = \corank \Pi_T(x) $;
		\item The linearized pencil $\diff_\lambda \Pi_T(x)$ is non-degenerate.
		\end{itemize}
		The point $x$ is non-degenerate if and only if these conditions are satisfied for each $\lambda$ in the spectrum.
	 If $x$ is non-degenerate, determine its type by adding up the types of $\diff_\lambda \Pi(x)$.
\end{enumerate}
}
\subsection{Computation of the spectrum}
First of all, compute the corank of the pencil. Note that all brackets of the pencil possess a common Casimir function
$$
C = \sum \log a_i,
$$
so $\corank \Pi \geq 2$. On the other hand, $\corank P_\infty = 2$, so $\corank \Pi_T = 2$.\par
Let $x = (a,b) \in M_T$.
Following \cite{GZToda}, consider the map $\mathcal{T}_\lambda \colon M_T \to M_T$ sending $(a,b)$ to $(a, b - \lambda)$. Clearly, $\mathcal{T}_\lambda$ maps the bracket $P_\lambda$ to the bracket $P_0$. So, $\lambda \in \Lambda(x)$ if and only if the rank of $P_0$ drops at $\mathcal{T}_\lambda(x)$. This observation reduces the study of singularities of the pencil to the singularities of $P_0$. \par
The kernel of $P_0$ consists of infinite $n$-periodic sequences $\alpha, \beta \in \R^{\infty}$ satisfying difference equations
\begin{align}\label{diffEq}
\begin{cases}
a_{i+1}\alpha_{i+1} -  a_{i-1}\alpha_{i-1} - 2b_i\beta_i +  2b_{i+1}\beta_{i+1} = 0,\\
a_ib_i\alpha_i -a_{i-1}b_i\alpha_{i-1} + 2a_i^2\beta_{i+1} -2a_{i-1}^2\beta_{i-1} = 0.
\end{cases}
\end{align}
Clearly, the space of all (not necessarily periodic) solutions of these equations is $4$-dimensional, so $\dim \Ker P_0 \leq 4$.  On the other hand, $\dim \Ker P_0 \geq 2$. Consequently, $\dim \Ker P_0$ is $2$ or $4$. Singular points of $P_0$ are exactly those where $\dim \Ker P_0 = 4$.
To describe these points, 
consider the infinite Lax matrix
\begin{align*}
\mathcal L(a,b) = \left(\begin{array}{cccccc}\ddots & \ddots  & \ddots &  &  &  \\ & a_{i-1} & b_i & a_{i} &  &  \\ &  & a_{i} & b_{i+1} & a_{i+1} &  \\ &  &  &  \ddots & \ddots & \ddots\end{array}\right),
\end{align*}
and the equation 
\begin{align}\label{LaxEq}\mathcal L(a,b)\,\xi = 0,\end{align} where  $\xi \in \R^\infty$. This equation can be written as an infinite sequence of difference equations
$$
a_{i-1}\xi_{i-1} + b_i\xi_{i} + a_{i}\xi_{i+1} = 0.
$$
Denote the space of  solutions of (\ref{LaxEq}) by $\Ker \mathcal{L}$.
Clearly, $\Ker \mathcal{L}$ is two-dimensional.
\begin{statement}
Assume that $\xi,\eta \in \Ker \mathcal L$. Then 
\begin{align}\label{prodRule}
\begin{cases}
\alpha_i = \xi_i\eta_{i+1} + \xi_{i+1}\eta_i,\\
\beta_i =  \xi_i\eta_i
\end{cases}
\end{align}
is a solution of (\ref{diffEq}).

\end{statement}
The proof is a straightforward computation.\par
\begin{remark}
A similar observation regarding the KdV equation is made in \cite{GZ2}: the product of two solutions of the Hill equation belongs to the kernel of the Magri bracket.
\end{remark}
The solution of (\ref{diffEq}) given by (\ref{prodRule}) will be denoted by $\xi \times \eta$. 
Let $\xi, \eta$ be a basis in $\Ker \mathcal L$. Then $\xi \times \xi$,  $\eta \times \eta$, and $ \xi \times \eta$ are three solutions of (\ref{diffEq}). A fourth solution is the differential of the common Casimir $C$. It is easy to check that the four solutions
$$
 \xi \times \xi, \quad \eta \times \eta, \quad \xi \times \eta, \quad \diff C
$$ 
are independent and thus span the space of solutions of (\ref{diffEq}). Obviously, all these four solutions are periodic if and only if $\xi$ and $\eta$ are either both periodic, or both anti-periodic. Consequently, $\dim \Ker P_0 = 4$ if and only if the equation $\mathcal L(a,b)\,\xi = 0$ has either two periodic, or two anti-periodic solutions, i.e. zero is a multiplicity-two periodic or anti-periodic eigenvalue of the $\mathcal L(a,b)$.

\begin{remark}
Note that $\dim \Ker P_0 \geq 2$, so the space of periodic solutions of (\ref{diffEq}) is at least two-dimensional. One periodic solution is $\diff C$. A second periodic solution can be constructed as follows. Let $\xi, \eta$ be two solutions of (\ref{LaxEq}). Then it is easy to see that the quantity 
\begin{align}\label{Wronskian}
\mathcal{W}(\xi, \eta) =\mathcal{W}_i(\xi, \eta)  = a_i(\xi_{i+1}\eta_i - \xi_i\eta_{i+1})
\end{align}
does not depend on $i$. Further, consider the operator on $\R^{\infty}$ which shifts a sequence to the left by $n$. Its restriction to $\Ker \mathcal L$ is the monodromy operator $$\mathcal{M} \colon \Ker \mathcal L \to \Ker \mathcal L.$$ Let $\xi, \eta$ be a basis in $\Ker \mathcal{L}$. Then
$$
\det \mathcal{M} = \frac{\mathcal{W}_{n+1}(\xi, \eta)}{\mathcal{W}_1(\xi, \eta)},
$$
so  $\mathcal M \in \mathrm{SL}(2, \R)$. There are two possible cases. \begin{enumerate} \item $\mathcal M$ is diagonalizable. Let $\xi, \eta$ be the eigenvectors of $\mathcal M$. Then the corresponding eigenvalues have product one, which implies that $\xi \times \eta$ is a periodic solution of  (\ref{diffEq}). \item $\mathcal M$ has a Jordan block with an eigenvalue $\pm 1$. Let $\xi$ be the eigenvector of $\mathcal M$. Then $\xi \times \xi$ is a periodic solution of  (\ref{diffEq}).
\end{enumerate}
\end{remark}

Now describe the spectrum and the singular set. Recall that $\lambda \in \Lambda(x)$ if and only if the rank of $P_0$ drops at $\mathcal{T}_\lambda(x)$. Further,
$$
\mathcal L(\mathcal{T}_\lambda(x)) = \mathcal L(x) - \lambda E,
$$
which implies the following.
\begin{theorem}\label{TodaSing}
Consider the Poisson pencil $\Pi_{T}$ associated with the periodic Toda lattice.
\begin{enumerate}
	\item The spectrum of the pencil at a point $x \in M_T$ coincides with the set of multiplicity-two periodic or anti-periodic eigenvalues of the Lax operator $\mathcal L(x)$.
	\item The singularities of the system are exactly those points where there are multiplicity-two periodic or anti-periodic eigenvalues of the Lax operator $\mathcal L(x)$.
\end{enumerate}
\end{theorem}
\begin{remark}
 Instead of considering the infinite Lax matrix $\mathcal L$, it can be restricted to the double period.  
\end{remark}
Since the Lax operator $\mathcal L$ is symmetric, we obtain the following.
\begin{consequence}
	The spectrum of the pencil $\Pi_T$ is real.
\end{consequence}

\subsection{Diagonalizability}
Let $x \in M_T$, and $\lambda \in \Lambda(x)$. In this section we compute the restriction of $P_{\infty}(x)$ to the kernel of $P_\lambda(x)$, and hence prove the diagonalizability of the pencil. \par According to the previous section, we may assume that $\lambda = 0$. The kernel of $P_0(x)$ is generated by
$$
 \xi \times \xi, \quad \eta \times \eta, \quad \xi \times \eta, \quad \diff C,
$$ 
where $\xi$ and $\eta$ are two periodic or anti-periodic solutions of (\ref{LaxEq}). 
\newpage Since $C$ is a common Casimir, we have $\diff C \in \Ker P_0$. A simple computation shows that
\begin{align*}
	P_\infty(\xi \times \xi, \eta \times \eta) &= 4\mathcal{W}(\xi, \eta)\sum\limits_{i=1}^{n} \xi_i\eta_i, \quad
	P_\infty(\xi \times \eta, \xi \times \xi) = -2\mathcal{W}(\xi, \eta)\sum\limits_{i=1}^{n} \xi_i^2,\\
	&P_\infty(\xi \times \eta, \eta \times \eta) = 2\mathcal{W}(\xi, \eta)\sum\limits_{i=1}^{n} \eta_i^2,
\end{align*}
where $\mathcal{W}(\xi, \eta)$ is given by (\ref{Wronskian}). Introduce the scalar product 
$$
\langle \xi, \eta \rangle = \sum\limits_{i=1}^{n} \xi_i\eta_i,
$$
and let $\xi, \eta$ be an orthonormal basis in $\Ker \mathcal L$. Then
$$
P_\infty(\xi \times \xi, \eta \times \eta) =0, \quad P_\infty(\xi \times \eta, \xi \times \xi) = -2\mathcal{W}(\xi, \eta), \quad P_\infty(\xi \times \eta, \eta \times \eta) = 2\mathcal{W}(\xi, \eta),
$$
so the kernel of $P_\infty$ restricted to $\Ker P_0(x)$ is generated by $\diff C$ and $\xi \times \xi + \eta \times \eta$, which proves that the pencil is diagonalizable.

\subsection{Linearization}
Let $x \in M_T$, and $\lambda \in \Lambda(x)$. In this section we compute the $\lambda$-linearization of the pencil and hence prove that all singularities are non-degenerate elliptic. Analogously to the previous section, we may assume that $\lambda = 0$.\par
Compute the Lie structure of $\Ker P_0(x)$. Let $(\alpha, \beta), ( \alpha',  \beta') \in \Ker P_0(x)$. By definition, their commutator is computed to be
$( \alpha'',  \beta'')$, where
\begin{align*}
 \alpha''_i &= \frac{1}{2}a_{i-1}(\alpha_i\alpha_{i-1}' -\alpha'_i\alpha_{i-1} ) + \frac{1}{2}a_{i+1}(\alpha_{i+1}\alpha_{i}' -\alpha'_{i+1}\alpha_{i} )  + b_i(\alpha_i\beta'_i - \alpha'_i\beta_i) \,+ \\&\qquad \quad \qquad+ \,b_i(\alpha'_i\beta_{i+1} - \alpha_i\beta_{i+1}') + 2a_i(\beta_{i+1}\beta_i' - \beta_{i+1}'\beta_i),
 \end{align*}
 and 
 \begin{align*}
\beta''_i = a_i(\alpha_i\beta'_i - \alpha'_i\beta_i) + a_{i-1}(\alpha'_{i-1}\beta_{i} - \alpha_{i-1}\beta_{i}').
\end{align*}
Using formulas (\ref{prodRule}) for the $\times$-product, we obtain
$$
[\xi \times \xi, \eta \times \eta] = 2\mathcal{W}(\xi, \eta) \,\xi \times \eta, \quad [\xi \times \eta, \xi \times \xi] = - 2\mathcal{W}(\xi, \eta) \,\xi \times \xi, \quad [\xi \times \eta, \eta \times \eta] =  2\mathcal{W}(\xi, \eta) \,\eta \times \eta.
$$
Further, $C$ is a Casimir, so $\diff C$ belongs to the center of $\Ker P_0$. Consequently,
$$
\Ker P_0(x) \simeq \sL(2, \R) \oplus \R.
$$
Let $\xi, \eta$ be an orthonormal basis in $\Ker \mathcal{L}$. Then the kernel of $P_\infty$ restricted to $\Ker P_0(x)$ is generated by $\diff C$ and $\xi \times \xi + \eta \times \eta$. Clearly, these two elements span a Cartan subalgebra in $\Ker P_0(x)$, so the linearization $\diff_0 \Pi(x)$ is non-degenerate. Moreover, the eigenvalues of the operator $\ad(\xi \times \xi + \eta \times \eta)$ are pure imaginary, so $\mathrm{Sing}(\diff_0 \Pi(x))$ is an elliptic singularity. We conclude with the following.
\begin{theorem}
All singularities of the periodic Toda lattice are non-degenerate elliptic.
\end{theorem}
\begin{consequence}
	All, including singular, trajectories of the periodic Toda lattice with $n$ sites evolve on $m$-dimensional tori, $m \leq n-1$. These tori are stable in a sense that if two initial conditions are close, so are the tori on which they lie. 
\end{consequence}
\begin{remark}
The above analysis suggests the similarity between the geometry of the quadratic Toda bracket and the linear bracket on the Virasoro dual. In particular, the transverse Poisson structure at singular points is $ \sL(2, \R) \oplus \R$ for both brackets. 
\end{remark}
\section{Proofs}\label{proofsSection}


\subsection{Space $L$ and recursion operators}\label{zeroOrderTheory}

In this section we study properties of two compatible Poisson brackets at a point, i.e. properties of a pair of skew-symmetric bilinear forms on a vector space. Note that all of these properties can be easily deduced from the Jordan-Kronecker theorem (see Appendix A). \par 
Consider a pencil $\Pi=\{P_\alpha\}$, and let  $x\in M$ be such that $\rank \Pi(x)=\rank \Pi$. Consider the spectrum of $\Pi$ at the point $x$:
$$
\Lambda = \{ \lambda \in \CP: \rank P_{\lambda}(x) < \rank \Pi(x) \},
$$
and define the subspace of $\T^*_xM$ spanned by the kernels of regular brackets:
$$
L = \sum\limits_{\lambda \in \overline \R \setminus \Lambda} \Ker P_{\lambda}(x) \subset \T^{*}_{x}M.
$$

\begin{statement}[Properties of the space $L$]\label{PropertiesOfL}
\quad\par
\begin{enumerate}
       \item The space $L$ is isotropic with respect to any bracket of the pencil.
       \item  The skew-orthogonal complement to $L$ given by $L^\bot = \{ \xi  \in \T^{*}_{x}M\mid P_\alpha(\xi,L) = 0\}$ does not depend on the choice of $\alpha\in\overline\R$.
       \item Any regular bracket of the pencil is non-degenerate on $L^\bot /L$.
       \item Let $k \geq \dim L$. Then for any distinct $\alpha_{1}, \dots, \alpha_{k} \in \overline \R \setminus \Lambda$ the following equality holds
	$$
		\sum\limits_{i=1}^{k} \Ker P_{\alpha_{i}} = L.
	$$
	\item $\dim (\Ker P_{\lambda} \cap L) = \corank \Pi(x)$ for all $\lambda\in \RP$. \item Similarly, $\dim_\Complex (\Ker P_{\lambda} \cap L \otimes \Complex) = \corank \Pi(x)$ for all $\lambda\in \Complex \setminus \R$.
	\end{enumerate}
\end{statement}

\begin{consequence}\label{dFisL}
	Let $\F$ be the system of functions defined in Section \ref{ihs}. Then $\diff \F = L$.
\end{consequence}

\begin{proof}
	By construction, $\F$ is generated by local Casimir functions of an infinite number of regular brackets. So,
\begin{align}\label{dfsum}
\diff \F = \sum\limits_{\alpha \in I} \Ker P_{\alpha},
\end{align}
where $I \subset \RP \setminus \Lambda$ is an infinite set. Consequently, $\diff \F \subset L$. On the other hand, by item 4 of Proposition \ref{PropertiesOfL}, any $k=\dim L$ summands of (\ref{dfsum}) generate $L$, so $\diff \F = L$.
\end{proof}

	Since $L$ is isotropic w.r.t. $P_{\beta}$,  the form  $P_\beta$ is well defined on $L^\bot /L$. Moreover, if $\beta$ is regular, i.e., $\beta\notin \Lambda$, then $L$ contains the kernel of $P_\beta$ and therefore $P_\beta$ on $L^\bot /L$ is non-degenerate.    Thus, the recursion operator
	$$R_{\alpha}^{\beta} = P_{\beta}^{-1}P_{\alpha} \colon L^\bot/L \to L^\bot/L$$
	 is well-defined.

\begin{statement}[Properties of recursion operators]\label{PropertiesOfR}
	\quad\par
	\begin{enumerate}
	\item For any $\alpha,\gamma$ and regular $\beta, \delta$ we can find constants $a,b,c,d$ such that
	\begin{align*}
		R_{\alpha}^{\beta}  = (aR_{\gamma}^{\delta} + b \Id)^{-1} (cR_{\gamma}^{\delta} + d\Id).
	\end{align*}
	Consequently, the operators $R_{\alpha}^{\beta}$ and $R_{\gamma}^{\delta}$ commute and have common eigenspaces. If one of the recursion operators is diagonalizable, then all of them are diagonalizable.
	\item Let $P_\infty\in\Pi$ be regular at $x$. Then the spectrum of the recursion operator $R_0^\infty$ is minus the spectrum of the pencil, i.e.
	$$
	 \sigma(R_0^\infty) = \{ -\lambda\}_{\lambda \in \Lambda}.
	$$
	The $\lambda$-eigenspace of $R_0^\infty$ is
	$$
		\Ker (R_0^\infty - \lambda \Id) = \Ker \left(P_{-\lambda}\mid_{L^\bot / L}\right).
	$$
	\item The eigenspaces of a recursion operator are pairwise orthogonal with respect to all brackets of the pencil.
	\item A pencil is diagonalizable at the point $x$ (in the sense of Definition \ref{diagpencil}) if and only if the recursion operators are $\Complex$-diagonalizable.
	\end{enumerate}
\end{statement}

The proof is straightforward.


\subsection{Operator $\D_{f}P$ and non-degeneracy}\label{firstOrderTheory}

Since in our paper we deal with Poisson brackets, it will be convenient to reformulate the standard  definition of non-degenerate singularities (Definition \ref{nd}), given in terms of the symplectic form, in the dual language of Poisson structures. To that end, in this section we replace linearizations of Hamiltonian vector fields by dual operators acting on the cotangent space and study properties of such operators.

 Let $P$ be a Poisson bracket on $M$, $x \in M$, and $f$ be a function such that $\diff f(x) \in \Ker P(x)$.
	Define $\D_{f}P\colon \T^{*}_{x}M \to \T^{*}_{x}M$ by the following formula
	$$
		 \D_{f}P(\xi)  = \diff \{ f, g \}(x),
	$$
	 where $g$ is an arbitrary function such that $\diff g(x) = \xi$.

	\begin{statement}[Properties of $\D_f P$]\label{dfpProperties}
	\quad\par
	\begin{enumerate}
		\item In local coordinates $\D_f P$ reads $$ (\D_{f}P(\xi))_{k}  = \diffFX{P^{ij}}{x^{k}} \diffFX{f}{x^i} \xi_{j} + P^{ij} \diffFXY{f}{x^{i}}{x^{k}} \xi_{j},$$ and therefore $\D_{f}P(\xi)$ does not depend on the choice of $g$.
		\item $\D_{f}P$ is dual to the linearization of the vector field $\sgrad f$ at the point $x$.
		\item $\D_{f}P$ is skew-symmetric with respect to the form $P$, i.e.
$$
P(\D_{f}P(\xi), \eta) + P(\xi, \D_{f}P(\eta)) = 0.
$$
		\item $\Ker P(x)$ is invariant with respect to $\D_{f}P$.
		\item Let $\xi \in \Ker P(x)$. Then
		 $$
		 	 \D_{f}P(\xi)  = [\diff f(x), \xi],
		$$
		where $[~,~]$ is the bracket on $\Ker P(x)$ viewed as the linearization of $P$ at $x$ (see Section \ref{linpen}). Equivalently, $\D_f P\mid_{\Ker P(x)} = \mathrm{ad}\,{\diff f(x)}$.
		\item If $\rank P(x) = \rank P$, then $ \D_{f}P$ vanishes on $\Ker P(x)$.
				\item If $\diff f(x)=0$, then $\D_f P(\xi) = \diff^2f(P(\xi))$.
	\end{enumerate}
	\end{statement}
	The proof is straightforward.

 If we consider $P$ as a map $P:\T^*_x M \to  \T_x O(x)$, then
 items 2, 3 and 4 of Proposition \ref{dfpProperties} imply  the commutativity of the diagram
\begin{align}\label{commdiag}
       				 \begin{CD}
           			 	\T^*_x M @> \D_{f}P >> \T^*_x M\\
            				@VV P V  @VV P V \\
          				 \T_x O(x) @> A_{f} >>\T_xO(x)
       			 \end{CD}
   			 \end{align}
where $A_f$ is the linearization of $\sgrad f$ at $x\in M$.

This simple observation allows us to reformulate Definition \ref{nd} of a non-degenerate point in terms of operators $\D_{f}P$. We simply need to modify the scheme preceding Definition \ref{nd}  by replacing $A_f$ with $\D_{f}P$ and passing to the dual objects.

Let $\F\subset \Cont^\infty (M)$ be a commutative subalgebra, which defines an integrable system on the symplectic leaf $O(x)$ passing through $x\in M$.  For simplicity, we assume that $O(x)$ is regular, i.e., $\rank P(x)=\rank P$.

Consider the subspace
$$
V=\diff \F(x)=\{ \diff f(x)\}_{f\in\F} \subset \T^*_xM.
$$
Since $O(x)$ is regular then, without loss of generality, we may assume that $\F$ contains (local) Casimirs of $P$ so that $\Ker P \subset V$.  If not, we simply set $V=\diff \F(x)+\Ker P$ and then the construction below can also  be applied to singular symplectic leaves.

Consider the skew-orthogonal complement to $V$ with respect to $P$, i.e. the subspace
$V^\bot = \{ \xi \in \T^*_x M ~|~  P(\xi, V)=0\}$. Since $\F$ is commutative,  then $V$ is isotropic and $V\subset V^\bot$. Moreover, since $V$ contains $\Ker P$, the quotient $V^{\bot}/V$ carries a natural symplectic form induced by $P$.
\newpage
Similarly, set
$$
W = \{\sgrad f(x) = P\diff f(x)\}_{f \in \F} \subset T_xO(x),
$$
and let $W^\bot$ be the skew-orthogonal complement to $W$ with respect to the symplectic form $\omega$ on $O(x)$ induced by $P$. As $\F$ is commutative, then $W$ is isotropic and the space $W^\bot / W$ is symplectic with respect to $\omega$.

Now let  $f \in \F$ and $\diff f(x) \in \Ker P$. Since $\F$ is commutative, $\D_fP$ vanishes on $V$, and item 3 of Proposition \ref{dfpProperties} implies that $V^\bot$ is invariant under  $\D_fP$.  Consequently, the operator $\D_fP$ is well-defined on the space $V^{\bot} / V$. Similarly, the linearization $A_f$ is a well-defined operator on $W^{\bot}/W$.

	\begin{statement}\label{isom}
		 Consider the map $P: \T^{*}_{x}M \rightarrow \T_{x} O(x)$. The following is true:
		 \begin{enumerate}
		 	\item  $\omega(P(\xi), P(\eta)) = P(\xi, \eta), \  \xi,\eta\in \T^*_x M.$
		 	\item $P(V) = W$.
			 \item $P(V^{\bot}) = W^{\bot}$.
             \item $P$ induces an isomorphism between the symplectic spaces $(V^\bot/V,P)$ and $(W^\bot/W,\omega)$.
			 \item Let  $f \in \F$, $\diff f(x) \in \Ker P$. Then the following diagram is commutative:
			     \begin{align*}
       				 \begin{CD}
           			 	V^{\bot} / V @> \D_{f}P >> V^{\bot} / V\\
            				@VV P V  @VV P V \\
          				 W^\bot / W @> A_{f} >>W^\bot / W
       			 \end{CD}
   			 \end{align*}
		 \end{enumerate}
	\end{statement}

The proof easily follows from Proposition \ref{dfpProperties} and can be understood as a reduction of the commutative diagram \eqref{commdiag}.

    \begin{consequence}\label{precriterion}
        A singular point $x$ of $\F$  is non-degenerate on the regular symplectic leaf $O(x)$ of $P$ passing through $x$ if and only if the set of operators
        $$\D_\F=\left\{ \D_{f}P\mid_{V^\bot/V}  \right\}_{f \in \F, \diff f \in \Ker P}$$
         is a Cartan subalgebra of $\sP( V^{\bot} / V, P)$. The type of the point $x$ coincides with the type of this Cartan subalgebra.
     \end{consequence}

    \begin{proof}
    	According to Definition \ref{nd}, $x$ is non-degenerate if and only if the linearizations $A_f$ of the Hamiltonian vector fields $\sgrad f$, where $f \in \F$ and $\sgrad f(x)=0$,  span a Cartan subalgebra in $\sP(W^\bot / W, \omega)$. Now we apply the isomorphism and commutative diagram constructed above in Proposition \ref{isom}.
    \end{proof}


\subsection{Operator $\D_{f}P$ and bi-Hamiltonian systems}

In this section, we consider the Poisson pencil $\Pi=\{P_\lambda\}$ and the corresponding family of commuting functions $\F$ from Section \ref{ihs}.  Let   $O(x,\alpha)$ be a regular symplectic leaf of $P_\alpha\in\Pi$, i.e. $\rank P_\alpha(x)=\rank \Pi$.   By  Corollary \ref{precriterion},  the non-degeneracy of $x$ can be formulated in terms of the commutative family of operators $\D_f P_\alpha$ for $f\in \F$.  Notice that in this case, the subspace $V=\diff F(x)$ coincides with $L$ defined in Section \ref{zeroOrderTheory}  (Proposition \ref{PropertiesOfL}) so that
\begin{equation}
\label{DF}
\D_\F=\left\{ \D_{f}P_\alpha\mid_{L^\bot/L}  \right\}_{f \in \F, \diff f \in \Ker P_\alpha}
         \subset \sP( L^{\bot} / L, P_\alpha).	
\end{equation}

In this section, we shall see that $\D_\F$ possesses some very special algebraic properties.

The following lemma allows us to rewrite the operator $\D_fP_\alpha, f \in \F$ as $\D_g P_\lambda$ for an appropriate function $g\in\F$. In particular, this implies that in \eqref{DF} the choice of $\alpha$ is not important (see Corollary \ref{dFPalphaInd}).

\begin{lemma}\label{biham}
        Let $\alpha_1, \dots, \alpha_k \in \Complex$, and  $f_{\alpha_i}$ be an arbitrary Casimir function of $P_{\alpha_i}$, $i=1,\dots, k$. Consider
        $$f = \sum\limits_{i=1}^{k} f_{\alpha_i} ,$$ and assume that $\diff f(x) \in \Ker P_\alpha(x)$. Let also $\beta \in \Complex \setminus \{\alpha_1, \dots, \alpha_k\}$ and consider the function
        $$
            g =  \sum\limits_{i=1}^{k} \frac{\alpha - \alpha_i}{\beta - \alpha_i} f_{\alpha_i}.
        $$
\newpage
        Then
        \begin{enumerate}
        \item $\diff g(x) \in \Ker P_\beta(x)$,
        \item $\D_fP_\alpha = \D_g P_\beta$.
        \end{enumerate}
\end{lemma}

The proof is a straightforward computation.

\begin{consequence}\label{dFPalphaInd}
	Let $f \in \F$ and $\alpha, \beta \in \CP$. Then there exists $g \in \F$ such that $\D_fP_\alpha = \D_g P_\beta$.
\end{consequence}
\begin{proof}
	Let  $$f = \sum\limits_{i=1}^{k} f_{\alpha_i} ,$$ where $f_{\alpha_i}$ is a Casimir function of $P_{\alpha_i}$. If $\beta \neq \alpha_i$ for any $i$, then the statement follows from Lemma \ref{biham}. Assume that $\beta = \alpha_1$. Let $t \to \alpha_1$ and let $f_{t}$ be a Casimir function of $P_{t}$ depending smoothly on $t$. Consider
	$$
	f(t) = f_t + \sum\limits_{i=2}^{k} f_{\alpha_i}.
	$$
	By Lemma \ref{biham}, there exists $g(t) \in \F$ such that 
	$$
	\D_{f(t)}P_\alpha = \D_{g(t)} P_\beta.
	$$
	Consequently,
	$$
	 \D_fP_\alpha = \lim_{t\to \alpha_1} \D_{f(t)}P_\alpha = \lim_{t\to \alpha_1} \D_{g(t)}P_\beta,
	$$
	so $ \D_fP_\alpha$ belongs to the closure of $ \{  \D_{h}P_\beta\}_{h \in \F}$. Now note that $ \{  \D_{h}P_\beta\}_{h \in \F} \subset \mathrm{End}(\T_x^*M)$ is a finite-dimensional vector space, so it is closed. Therefore, $\D_fP_\alpha = \D_g P_\beta$ for some $g \in \F$, q.e.d.

\end{proof}

\begin{consequence}\label{skewSymm}
        Let $f \in \F$ and $\diff f(x) \in \Ker P_\alpha$. Then $\D_f P_\alpha$ is skew-symmetric with respect to all brackets of the pencil.
\end{consequence}

\begin{proof}
It is enough to show that $\D_f P_\alpha$ is skew-symmetric with respect to two brackets of the pencil. By Proposition \ref{dfpProperties}, it is skew-symmetric with respect to $P_\alpha$.
	By Lemma \ref{biham} we can find $\gamma \neq \alpha$ and a function $g$ such that
	$
            \D_fP_\alpha =\D_{g} P_\gamma
        $.
        Therefore, $\D_f P_\alpha$ is skew-symmetric with respect to $P_{\gamma}$ as well, q.e.d.
\end{proof}


	It follows from the previous section that  the operator $\D_fP_\alpha$ is well-defined on the space $L^{\bot} / L$ for any $f \in \F$, $\diff f\in \Ker P_\alpha$ and moreover $\D_fP_\alpha\mid_{L^{\bot} / L}$ belongs to the symplectic Lie algebra $\sP(L^{\bot} / L, P_{\alpha})$.  As was shown in  Section \ref{zeroOrderTheory},   each $P_\beta$  induces a skew-symmetric form on $L^{\bot} / L$ and if $P_\beta$ is regular,  then this form is non-degenerate. Thus, $L^\bot/L$ is endowed with a pencil of skew-symmetric forms $\{ P_\lambda\mid_{L^\bot/L}\}$ and almost all of them are non-degenerate. Thus, it makes sense to  introduce the bi-symplectic Lie algebra associated with the pencil $\Pi$ by setting
$$
\sP(L^\bot/L, \Pi)=\bigcap_\beta \sP(L^\bot/L,P_\beta) = \sP(L^\bot/L,P_{\beta_1})\cap \sP(L^\bot/L,P_{\beta_2}),
$$    	
where $P_\beta, P_{\beta_1}$ and $P_{\beta_2}$ are regular in $\Pi$, $\beta_1\ne\beta_2$.

 \begin{statement}\label{dfpOnQuotient}
The commutative subalgebra $\D_\F$ defined by \eqref{DF} possesses the following properties.
            \begin{enumerate}
            \item $\D_\F \subset \sP(L^{\bot} / L, \Pi)$.
            \item $\D_\F$ commutes with the recursion operators, i.e., $[D,R_\beta^\alpha]=0$ for all $D\in \D_\F$.
            \item Common eigenspaces of the recursion operators are invariant under $\D_\F$.
            \item For any  $D\in \D_\F$, we have $D R_\beta^\alpha \in \D_\F$, i.e. $\D_\F$ is invariant under multiplication by $R_\beta^\alpha$.
                \label{st4}
            \end{enumerate}
 \end{statement}

 \begin{proof}
 	Items 1--3 are straightforward. Let us prove item 4. Since all the recursion operators are rational functions of each other, the choice of $\alpha$ is not important, and we may assume that $\alpha$ is such that all Casimir functions of $P_\alpha$ belong to $\F$.   Let $D=\D_{f }P_\alpha\mid_{L^{\bot} / L}$ for some $f\in\F$, $\diff f \in\Ker P_\alpha$. Take a Casimir function $f_\alpha$ of $P_\alpha$ such that $\diff f_\alpha = \diff f$. Then $\diff(f - f_\alpha) = 0$ and
	$$
	D = \D_fP_\alpha\mid_{L^{\bot} / L} = \D_{f - f_\alpha}P_\alpha\mid_{L^{\bot} / L} .
	$$
	On the other hand, by item 7 of Proposition \ref{dfpProperties},
	$$
	\D_{f - f_\alpha}P_\alpha\mid_{L^{\bot} / L} = \diff^2(f - f_\alpha)\left(P_\alpha\mid_{L^{\bot} / L}\right).
	$$
	Therefore,
	$$
	 DR_\beta^\alpha =  \diff^2(f - f_\alpha)\left(P_\alpha\mid_{L^{\bot} / L} R_\beta^\alpha \right) = \diff^2(f - f_\alpha)\left(P_\beta\mid_{L^{\bot} / L}\right) = \D_{f - f_\alpha}P_\beta\mid_{L^{\bot} / L},
	 	$$
		so $D R_\beta^\alpha \in \D_\F$, q.e.d.
%
 \end{proof}

The following Proposition allows us to calculate $\D_fP_\alpha$ on an eigenspace of the recursion operator.

\begin{statement}\label{dfpIsAd}\quad\par
	 \begin{enumerate}
	 \item
	 Let $f = \sum\limits_{i=1}^{k} f_{\alpha_i}$, where $f_{\alpha_i}$ is a Casimir function of a regular bracket $P_{\alpha_i}$.  Let also $\diff f(x) \in \Ker P_\alpha$, and $\lambda \in \Lambda(x)$. Then
 $\D_fP_\alpha\mid_{\Ker P_{\lambda}} = \ad\,\xi$, where
	 $$
            		\xi  = \sum\limits_{i=1}^{k} \frac{\alpha - \alpha_i}{\lambda - \alpha_i} \diff f_{\alpha_i},
	 $$
	 and $\ad\,\xi$ is the adjoint operator on $\g_{\lambda}=\Ker P_\lambda$.
	 \item	 The following sets of operators are equal
	 $$
	 	\left\{ \D_fP_\alpha\mid_{\Ker P_{\lambda}}\right\}_{f \in \F, \diff f \in \Ker P_{\alpha}} = \{ \ad\,{\xi}\}_{\xi \in \g_{\lambda} \cap L},
	$$
	where $\ad\,{\xi}$ is the adjoint operator on $\g_{\lambda}=\Ker P_\lambda$.
	\end{enumerate}
\end{statement}

\begin{proof}
	This directly follows from Proposition \ref{dfpProperties} and Lemma \ref{biham}.
\end{proof}
\begin{remark}
	If $\lambda$ is complex, $L$ should be replaced by $L \otimes \Complex$ in item 2.
\end{remark}

Notice that if the recursion operators are diagonalizable, we are now able to express $\D_fP_\alpha$ on the whole $L^{\bot} / L$ via adjoint operators.


    \subsection{Regular cocycles and central extensions}\label{race}

    In this section we give a condition for a $2$-cocycle to be regular in terms of the associated central extension.
	Let $\g$ be a Lie algebra over a field $\K$. Suppose that $A$ is a $2$-cocycle on $\g$. Let us consider the space $\g_{A} = \g + \K^{1}$, where $\K^{1} = \langle z \rangle$ is a one-dimensional vector space, and define a commutator $\LieBracket_{A}$ on $g_{A}$ by the following rule:
\begin{align*}
	[x,y]_{A} = [x,y] + A(x,y)z&, \mbox{ for any } x,y \in \g \subset \g_{A},\\
	[z,\g_{A}]_{A} &= 0.
\end{align*}
It is easy to see that if $A$ is closed, then the commutator $\LieBracket_{A}$ turns  $\g_{A}$ into a Lie algebra. Also note that $\g = \g_{A} / \langle z \rangle$, and the lift of $A$ to $g_{A}$ is an exact form. This means that every closed $2$-form on a Lie algebra becomes exact after being lifted to a certain one-dimensional central extension.\par
The following is straightforward.

\begin{statement}\label{regularitycriterion}
	A $2$-cocycle $A$ on $\g$ is regular if and only if its lift to $\g_A$ is a coboundary of a regular element $a \in \g_A^{*}$.
\end{statement}

\begin{consequence}\label{abelianProp}
	If $A$ is a regular cocycle on $\g$, then $\Ker A$ is Abelian.
\end{consequence}

\begin{proof}
	Let $\wave A$ be the lift of $A$ to $\g_{A}$. Since $\wave A$ is exact and regular, $\Ker \wave A$ is Abelian. Let $\pi \colon \g_A \to \g$ be the natural projection. Then $\Ker A = \pi(\Ker \wave A)$. Consequently, $\Ker A$ is also Abelian, q.e.d.
\end{proof}

\begin{consequence}
	If $A$ is a regular cocycle on $\g$, then the set of operators \begin{align}\label{adAlegbra}\{ \ad\,{\xi}\mid_{\g / \Ker A} \}_{ \xi \in  \Ker A}\end{align} is an Abelian subalgebra of $\sP(\g / \Ker A, A)$.
\end{consequence}

It is proved in the next section that  $\Pi^{\g, A}$ is non-degenerate if and only if the set (\ref{adAlegbra}) is not just an Abelian subalgebra, but a Cartan subalgebra.


\subsection{Proof of Theorem \ref{ndlp1}}

    Let $\Pi^{\g, A}$ be an integrable linear pencil. Construct the system $\F$ (see Section \ref{ihs}) for this pencil. The origin is the singular point of $\F$ restricted to the symplectic leaf of $A$. By Corollary \ref{precriterion}, to check non-degeneracy and find the type of this point, we need to calculate the operators $\D_{f}A$ on $L^{\bot} / L$ for each $f \in F$ such that $ \diff f \in \Ker A$.\par
    By Proposition \ref{dfpIsAd}, we have
    $$
    	\left\{ \D_fA\mid_{\Ker P_{0}} \right\}_{f \in \F, \diff f \in \Ker A} = \{ \ad\,{\xi} \}_{\xi \in  \Ker P_0 \cap L},
    $$
    where $P_0$ is the Lie-Poisson bracket. Since $\Ker P_{0} = \g$, and $\Ker A = L$, we have
     $$
    	\left\{ \D_fA \right\}_{f \in \F, \diff f \in \Ker A} = \{ \ad\,{\xi} \}_{\xi \in \Ker A}.
    $$
    Since $L^{\bot} = \g$, we have
     $$
    	\left\{ \D_fA\mid_{L^{\bot} / L} \right\}_{ f \in \F, \diff f \in \Ker A} = \{ \ad\,{\xi}\mid_{\g / \Ker A} \}_{ \xi \in  \Ker A}.
    $$
    Taking into account Corollary \ref{precriterion}, this proves the following.

    \begin{lemma}\label{linPrecriterion}
    		An integrable linear pencil $\Pi^{\g, A}$ is non-degenerate if and only if the set of operators
		$$\{ \ad\,{\xi}\mid_{\g / \Ker A} \}_{ \xi \in  \Ker A}$$ is a Cartan subalgebra in $\sP(\g / \Ker A, A)$. The type of $\mathsf{Sing}(\Pi^{\g, A})$ coincides with the type of this subalgebra.
    \end{lemma}

    \begin{consequence}\label{adss}
    		If $\Pi^{\g, A}$ is non-degenerate, then $\Ker A$ consists of $\ad$-semisimple elements.
    \end{consequence}

  Assume that $\Ker A$ is a commutative subalgebra which consists of semisimple elements. Then all operators $\ad\,{\xi}, \xi \in \Ker A$ may be simultaneously diagonalized (over $\Complex$). So, we can consider the root decomposition of $\g$:
  $$
  	\g \otimes \Complex  = \Ker A  \otimes \Complex + \sum\limits_{i=1}^{n}(V_{\lambda_{i}} + V_{-\lambda_{i}}),
  $$
  where each $V_{\pm \lambda_{i}}$ is spanned by one common eigenvector corresponding to the eigenvalue $\pm \lambda(\xi)$. Eigenvalues enter in pairs because the operators $\ad\,\xi$ are symplectic.

Lemma \ref{linPrecriterion} implies the following.

  \begin{statement}\label{indRoots}
  	Let $\Pi^{\g, A}$ be an integrable linear pencil. If $\Ker A$ is diagonalizable, then the pencil is non-degenerate if and only if $\lambda_{1}, \dots \lambda_{n}$ are linearly independent as linear functions on $\Ker A$. Type of $\mathsf{Sing}(\Pi^{\g, A})$ is $(k_{e}, k_{h}, k_{f})$ where $k_{e}$ is the number of pure imaginary $\lambda_{i}$'s, $k_{h}$ is the number of real $\lambda_{i}$'s, and $k_{f}$ is the number of pairs of complex conjugate $\lambda_{i}$'s.
  \end{statement}

%

\begin{proof}[Proof of Theorem \ref{ndlp1}]
 Taking into account Corollary \ref{adss} and Proposition \ref{indRoots}, it suffices to show that if $\g$ admits decomposition (\ref{ndlp1decomp}) satisfying the conditions of Theorem \ref{ndlp1}, then $A$ is regular and $\Pi^{\g, A}$ is integrable.\par
To prove regularity, note that the conditions of the theorem  are still satisfied if we pass to the central extension associated with $A$. So, by Proposition \ref{regularitycriterion}, it is enough to consider the case when $A = A_a$. In Sections \ref{nlpProofComplex}, \ref{nlpProofReal} we show that the conditions of the theorem imply that $\g$ belongs to the list (\ref{ndadecomp_c}) or (\ref{ndadecomp}). For these Lie algebras, the dimension of a maximal Abelian $\ad$-diagonalizable subalgebra equals the index. Therefore, $a$ is a regular element of $\g^*$, q.e.d.\par
   Now note that integrability follows from non-degeneracy at the origin.
\end{proof}


\subsection{Proof of Theorems \ref{nd1} and \ref{type1}}\label{mainProof}

	By Corollary \ref{precriterion}, a singular point $x$ is non-degenerate on a regular symplectic leaf of a bracket $P_\alpha$ if and only if the set of operators
	$$\D_\F = \{ \D_{f}P_\alpha\mid_{L^{\bot} / L} \}_{ f \in \F, \diff f \in \Ker P_\alpha}$$
	 is a Cartan subalgebra in $\sP(L^{\bot} / L, P_\alpha)$. The type of the singular point coincides with the type of $\D_\F$.

   \begin{statement}\label{partDiag}
   	Suppose that a point $x$ is non-degenerate. Then the pencil $\Pi$ is diagonalizable at  $x$.
   \end{statement}

   \begin{proof}
    Indeed, since $\D_\F$ is a Cartan subalgebra in $\sP(L^{\bot} / L, P_\alpha)$, it contains a semisimple operator $D$ with distinct eigenvalues. On the other hand, by Proposition \ref{dfpOnQuotient}, any recursion operator commutes with $D$. Therefore, all recursion operators are diagonalizable. So, by Proposition \ref{PropertiesOfR}, the pencil is diagonalizable at $x$.
    \end{proof}
Without loss of generality, assume that $\infty \notin \Lambda(x)$. Denote
\begin{align*}
\Lambda_\R = \Lambda(x) \cap \R, \quad \Lambda_\Complex =  \Lambda(x) \cap (\Complex \setminus \R).
\end{align*}
    In the diagonalizable situation, the space $L^{\bot} / L$ is (symplectically) decomposed into the direct sum of common eigenspaces of the recursion operators (Proposition \ref{PropertiesOfR}):
    \begin{equation}
    \label{splitting}
	L^{\bot} / L = \bigoplus_{\substack{\lambda \in \Lambda_\R}}K_{\lambda} \,\,\oplus  \!
\bigoplus_{\substack{\{\lambda,\bar\lambda\} \in \Lambda_\Complex}}K_{\lambda,\bar\lambda},
    \end{equation}
where $K_\lambda$ is the kernel of the operator $R^\infty_0+\lambda  \Id$ if $\lambda\in\mathbb R$, and $K_{\lambda,\bar\lambda}$ is the kernel of $(R^\infty_0+\lambda  \Id)(R^\infty_0 + \bar\lambda \Id)$ for a pair of complex conjugate eigenvalues $\{\lambda,\bar\lambda\}\subset \Lambda_\Complex$.

Decomposition (\ref{splitting}) is determined by the pencil. All these spaces are pairwise skew-orthogonal w.r.t. $P_\alpha$, and the restriction of $P_\alpha$ on each of them is non-degenerate.

Notice that in the case of a diagonalized recursion operator, the bi-symplectic Lie algebra
$\sP\left( L^\bot/L, \Pi\right)$ naturally splits into the direct sum
$$
\bigoplus_{\substack{\lambda \in \Lambda_\R}}
\sP(K_{\lambda})
\,\,\oplus\!\bigoplus_{\substack{\{\lambda,\bar\lambda\} \in \Lambda_\Complex}}
\sP_\Complex(K_{\lambda,\bar\lambda}).
$$
This follows immediately from the fact that $\sP\left( L^\bot/L, \Pi\right)$ is the intersection of
$\sP\left( L^\bot/L, P_\alpha\right)$ with the centralizer of the recursion operator $R_0^\infty$ which is the direct sum
$$
\bigoplus_{\substack{\lambda \in \Lambda_\R}}
\mathfrak{gl}(K_{\lambda})
 \,\,\oplus \!
\bigoplus_{\substack{\{\lambda,\bar\lambda\} \in \Lambda_\Complex}}
\mathfrak{gl}_{\Complex}(K_{\lambda,\bar\lambda})
.
$$

Here $\mathfrak{gl}_\Complex(K_{\lambda,\bar\lambda})$ denotes the algebra of operators that commute with the natural complex structure $J:  K_{\lambda,\bar\lambda} \to K_{\lambda,\bar\lambda}$ defined by the restriction of the recursion operator $R_0^\infty$ onto $K_{\lambda,\bar\lambda}$. Namely, if $\lambda=\alpha+i\beta$, then this restriction has the form $\alpha \Id + \beta J$, where $J^2=-\Id$. Notice that the dimension of $K_{\lambda,\bar\lambda}$ is divisible by 4, and $\sP_\Complex(K_{\lambda,\bar\lambda})$ is isomorphic to $\sP(2n,\Complex)$.

These eigenspaces $K_\lambda$ and $K_{\lambda,\bar\lambda}$ are invariant with respect to the operators $D\in \D_\F$ (Proposition \ref{dfpOnQuotient}).     Denote
\begin{align*}
	\D_{\lambda} = \D_\F\mid_{K_\lambda} = \left\{  D\mid_{K_\lambda}, D\in\D_\F\right\}\subset \sP(K_\lambda),
\end{align*}
and, similarly,
\begin{align*}
	\D_{\lambda,\bar\lambda} = \D_\F\mid_{K_{\lambda,\bar\lambda}} =
\left\{  D\mid_{K_{\lambda,\bar\lambda}}, D\in\D_\F\right\}\subset \sP(K_{\lambda,\bar\lambda}).
\end{align*}

Notice that $\D_{\lambda,\bar\lambda}$ can also be considered as a subalgebra of the complex symplectic Lie algebra  $\sP_\Complex(K_{\lambda,\bar\lambda})$, i.e., $\D_{\lambda,\bar\lambda}\subset \sP_\Complex(K_{\lambda,\bar\lambda}) \subset \sP(K_{\lambda,\bar\lambda})$.

Using the decomposition \eqref{splitting}
we can define the natural projection
$$
\pi \colon \D_\F \to \bigoplus_{\substack{\lambda \in \Lambda_\R}}
\D_{\lambda} 
\,\,\oplus \!
\bigoplus_{\substack{\{\lambda,\bar\lambda\} \in \Lambda_\Complex}}
\D_{\lambda,\bar\lambda} 
.
$$

\begin{statement}\label{part0}
The projection $\pi$ is an isomorphism.

    \end{statement}

    \begin{proof}
    Clearly, $\pi$ is injective, so it is enough to show that $\pi(\D_\F) \supset \D_{\lambda}$ and
      $\pi(\D_\F) \supset \D_{\lambda,\bar\lambda}$ for each $\lambda \in \Lambda(x)$. We prove this for real $\lambda$. For a pair of complex eigenvalues, the proof is similar.  
      
     Since $R_0^\infty$ is diagonalizable, we can find a polynomial $p_\lambda$ such that $p(R_0^\infty)$ is the projection onto $K_\lambda$. By, Proposition \ref{dfpOnQuotient}, $\D_\F$ is invariant under right multiplication by $R_0^\infty$, so
    $$
    \D_\F\, p_\lambda(R_0^\infty) \subset D_\F.
    $$
    On the other hand, by definition of $p_\lambda$,
    $$
    \D_\F\, p_\lambda(R_0^\infty) = \D_{\lambda},
    $$
    which proves the proposition.
    \end{proof}

\begin{statement}
\label{part1}
	Let $\Pi$ be diagonalizable at a point $x$, and  $\alpha \notin \Lambda(x)$. Then the singular point $x$ is non-degenerate on the symplectic leaf of $P_\alpha$ if and only if
$$
\D_{\lambda} \subset \sP\left( K_\lambda\right) \quad \text{and} \quad
\D_{\lambda,\bar\lambda} \subset \sP\left( K_{\lambda,\bar\lambda}\right)
$$
are Cartan subalgebras for all $\lambda \in \Lambda(x)$. The type of $x$ is the sum of the types of  $\D_{\lambda}$ and  $\D_{\lambda,\bar\lambda}$ over $\lambda \in \Lambda(x) \cap \{ \Imm \lambda \geq 0 \}$. 
 \end{statement}

 \begin{proof}
 	This follows from Proposition \ref{part0} and Corollary \ref{precriterion}.
 \end{proof}


 \begin{statement}\label{part2}
 Let $\Pi$ be diagonalizable at $x$ and $\lambda\in\Lambda_\R$. Then $\D_{\lambda}$ is a Cartan subalgebra in $\sP(K_\lambda)$ if and only if the pencil $\diff_{\lambda}\Pi(x)$ is non-degenerate. The type of $\D_{\lambda}$ coincides with the type of $\mathsf{Sing}(\diff_{\lambda}\Pi(x))$.
 \end{statement}

 \begin{proof}
 Recall that $\mathfrak g_\lambda = \Ker P_\lambda \subset L^\bot$ and $K_\lambda\subset L^\bot/ L$ is the image of $\goth g_\lambda$ under the natural projection of $L^\bot \to L^\bot/L$.
 	By Proposition \ref{dfpIsAd},
	$$
		 \left\{ \D_fP_\alpha\mid_{\mathfrak g_{\lambda}} \right\}_{f \in \F, \diff f \in \Ker P_{\alpha}} = \{ \ad\,{\xi} \}_{\xi \in \g_{\lambda} \cap L}.
	 $$
	By Proposition \ref{PropertiesOfL}, $\dim \left (\g_{\lambda} \cap  L\right) = \corank \Pi(x)$.
	Since $\Pi$ is diagonalizable at $x$, this implies that
	$$\g_{\lambda} \cap  L = \Ker A,\quad \text{where } A= P_\alpha \mid_{\g_\lambda}.$$ So,
		$$
		 \left\{ \D_fP_\alpha\mid_{\mathfrak g_{\lambda}} \right\}_{f \in \F, \diff f \in \Ker P_{\alpha}} = \{ \ad\,{\xi} \}_{\xi \in \Ker A},
	 $$
	 and, taking into account that $K_\lambda = \mathfrak g_\lambda/\Ker A$, we get
	 $$
\left\{ \D_fP_\alpha\mid_{K_{\lambda}} \right\}_{f \in \F, \diff f \in \Ker P_{\alpha}}=
\{ \ad\,{\xi}\mid_{\g_{\lambda} / \Ker A} \}_{\xi \in \Ker A}.
	 $$
	 Now it suffices to apply Proposition \ref{linPrecriterion}.
 \end{proof}

\begin{statement}\label{part3}
 	Let $\Pi$ be diagonalizable at $x$ and $\{\lambda,\bar\lambda\}\in\Lambda_\Complex$. Then $\D_{\lambda,\bar\lambda}$ is a Cartan subalgebra in $\sP( K_{\lambda, \bar\lambda})$ if and only if the pencil $\diff_{\lambda}\Pi(x)$ is non-degenerate.  In this case,  $\D_{\lambda,\bar\lambda}$ is of pure focus type.
 \end{statement}

 \begin{proof}

If $\lambda$ and $\bar\lambda$ are complex numbers, $\Imm \lambda\ne 0$, then the subspaces $K_\lambda$, $K_{\bar\lambda}$, $\D_\lambda$ and $\D_{\bar\lambda}$ still make sense but now all of them are related to the complexified space $(L^\bot/L)^\Complex=(L^\bot/L)\otimes\Complex$, namely,
$$
K_\lambda = \Ker \left(P_\lambda\mid_{(L^\bot / L) ^\Complex} \right)\subset (L^\bot / L) ^\Complex, \quad
K_{\bar\lambda} = \Ker \left(P_{\bar \lambda}\mid_{(L^\bot / L)^\Complex} \right)
\subset (L^\bot / L) ^\Complex$$
and
$$
\D_{\lambda} =
\left\{ D\mid_{K_\lambda}, D\in\D_\F\right\} \subset \sP( K_\lambda), \quad
\D_{\bar\lambda} =
\left\{ D\mid_{K_{\bar\lambda}}, D\in\D_\F\right\}\subset \sP( K_{\bar\lambda}).
$$

From the algebraic viewpoint, the relationship between these subspaces and subalgebras can be described as follows. 

With the vector space $K_{\lambda, \bar\lambda}$, we can associate two symplectic Lie algebras 
$\sP(K_{\lambda, \bar\lambda})$ and $\sP_\Complex(K_{\lambda, \bar\lambda})$ (they are isomorphic to $\sP(4k,\R)$  and $\sP(2k,\Complex)$ respectively).  The complexified vector space $K_{\lambda, \bar\lambda}^\Complex = K_{\lambda, \bar\lambda}\,\otimes \Complex$ naturally splits into the direct sum $K_{\lambda, \bar\lambda}^\Complex = K_\lambda \oplus K_{\bar\lambda}$.  This leads to the following  complexifications of the symplectic Lie algebras  
 $\sP(K_{\lambda, \bar\lambda})$ and $\sP_\Complex(K_{\lambda, \bar\lambda})$:
 $$
 \sP(K_{\lambda, \bar\lambda})\otimes\Complex = \sP(K_\lambda \oplus K_{\bar\lambda})\simeq \sP(4k,\Complex)
 $$
and 
$$
\sP_\Complex(K_{\lambda, \bar\lambda}) \otimes \Complex =
\sP (K_\lambda) \oplus \sP(K_{\bar\lambda}) \simeq \sP(2n,\Complex)\oplus \sP(2n,\Complex).
$$

We consider an Abelian subalgebra $\D_{\lambda,\bar\lambda}\subset \sP_\Complex(K_{\lambda, \bar\lambda}) \subset \sP(K_{\lambda, \bar\lambda})$ and are interested in conditions under which  
$\D_{\lambda,\bar\lambda}$ is a Cartan subalgebra of $\sP(K_{\lambda, \bar\lambda})$.

In our case, $\D_{\lambda,\bar\lambda}$ satisfies an additional important property of being invariant under left multiplication by the complex structure $J$, i.e., for every $D\in \D_{\lambda,\bar\lambda}$ we have $JD\in \D_{\lambda,\bar\lambda}$.    This immediately implies the following algebraic conclusion:   $\D_{\lambda,\bar\lambda}^\Complex = \D_{\lambda,\bar\lambda}\otimes\Complex = \D_\lambda \oplus \D_{\bar\lambda}$ and the following conditions are equivalent:

\begin{itemize}
\item $\D_{\lambda,\bar\lambda}$ is a Cartan subalgebra of $\sP(K_{\lambda, \bar\lambda})$,
\item $\D_{\lambda,\bar\lambda}$ is a Cartan subalgebra of $\sP_\Complex(K_{\lambda, \bar\lambda})$,
\item $\D_{\lambda}$ is a Cartan subalgebra of $\sP(K_{\lambda})$,
\item $\D_{\bar\lambda}$ is a Cartan subalgebra of $\sP(K_{\bar\lambda})$.
\end{itemize}

If one of these conditions is fulfilled then $\D_{\lambda,\bar\lambda}$ is of pure focus type as a Cartan subalgebra of the (real) symplectic Lie algebra $\sP(K_{\lambda, \bar\lambda})$.  
\newpage

To complete the proof it remains to notice that the statement of Proposition \ref{part2} still holds true for complex $\lambda\in\Lambda(x)$. The only difference is that the pencil $\diff_{\lambda}\Pi(x)$ is complex and all the Cartan subalgebras of $\sP(K_\lambda)$ are of the same type, i.e., are conjugate. \end{proof}

 \begin{proof}[Proof of Theorem \ref{nd1}]
 	By Proposition \ref{partDiag}, if a point $x$ is non-degenerate, then the pencil is diagonalizable at $x$. Therefore, it suffices to show that  for diagonalizable pencils,  $x$ is non-degenerate if and only if for each $\lambda \in \Lambda(x)$ the linear pencil $\diff_{\lambda}\Pi(x)$ is non-degenerate. This follows from Propositions \ref{part1}, \ref{part2}, \ref{part3}.
\end{proof}

\begin{proof}[Proof of Theorem \ref{type1}]
	The proof follows from Propositions \ref{part1}, \ref{part2}, \ref{part3}.
\end{proof}


  \subsection{Proof of Theorem \ref{nlpc}}\label{nlpProofComplex}

Taking into account Theorem \ref{ndlp1}, we need to prove the following.
\begin{lemma}\label{ndlpLemma}
Let $\g$ be a complex Lie algebra, and let $\h \subset \g$ be an Abelian subalgebra. Then the following two conditions are equivalent.
\begin{enumerate}
\item The Lie algebra $\g$ admits a decomposition
\begin{align*}
  	\g  = \h  + \sum\limits_{i=1}^{n}(V_{\lambda_{i}} + V_{-\lambda_{i}}),
  \end{align*}
where  \begin{enumerate} \item $\lambda_1, \dots, \lambda_n \in \h^*$ are linearly independent. \item Each space $V_{\pm \lambda_i}$ is one-dimensional. \item For all $\xi \in\h, e_{\lambda_i} \in V_{\lambda_i}$ we have $[\xi, e_{\lambda_i}] = \lambda_i(\xi)e_{\lambda_i}$. \end{enumerate}
\item The Lie algebra $\g$ can be represented as
		\begin{align}\label{ndAlgebraRepr}
		\g \simeq \bigoplus \so(3, \Complex) \oplus  \left(\bigoplus \foc \right) / \lCal_{0}  \oplus V,
	\end{align}	
	where $V$ is Abelian, $\lCal_{0}\subset \bigoplus \foc$ is a central ideal and $\h \subset \g$ is a Cartan subalgebra.
\end{enumerate}
\end{lemma}
\begin{proof}
%
The implication 2 $\Rightarrow$ 1 is straightforward. Prove 1 $\Rightarrow$ 2. \par Obviously, $\h$ is a Cartan subalgebra, so it suffices to prove that $\g$ admits representation (\ref{ndAlgebraRepr}).

Standard arguments show that if $e_{\alpha} \in V_{\alpha}$ and $e_{\beta} \in V_{\beta}$, then $[e_{\alpha}, e_{\beta}] \in V_{\alpha + \beta}$. Since the roots are independent, $\alpha + \beta$ is a root if and only if $\beta = -\alpha$. Consequently, the following relations  are satisfied.
\begin{align*}
	[V_{\lambda_{i}}, V_{-\lambda_{i}}] \in \h, \quad
	[V_{\lambda_{i}}, V_{\pm \lambda_{j}}] = 0 \mbox{ for } i \neq j.
\end{align*}
Let $e_{i}$ be a basis vector in $V_{\lambda_i}$ and $e_{-i}$ be a basis vector in $V_{-\lambda_i}$. Denoting $h_{i} = [e_{i}, e_{-i}]$ and using the Jacobi identity, we have
$$
	[h_{i}, e_{j}] = [[e_{i}, e_{-i}], e_{j}] = 0
$$
for $i \neq j$. Therefore,
$$
	\lambda_{j}(h_{i}) = 0 \mbox{ for } i \neq j.
$$
Now suppose that $\lambda_{i}(h_{i}) \neq 0$ for some value of $i$. Then the triple $e_{i}, e_{-i},h_{i}$ generates a subalgebra isomorphic to $\so(3, \Complex)$. Let us show that it admits a complementary subalgebra in $\g$.
Let
$$
\wave \h = \{h \in \h: \lambda_{i}(h) = 0\}.
$$
Denote
$$
\wave \g = \wave h + \sum\limits_{j \neq i} (V_{\lambda_{j}} + V_{-\lambda_{j}}).
$$
Then it is easy to see that $\g$ splits into the direct sum:
$$
\g = \wave \g \oplus \langle e_{i}, e_{-i},h_{i} \rangle.
$$

After separating $\so(3)$ summands for all $i$ such that $\lambda_{i}(h_{i}) \neq 0$, we may assume that $\lambda_{i}(h_{i}) = 0$ for all $i$. Now separate an Abelian summand. Decompose the center of $\g$ into a direct sum of $\zenter(\g) \cap [\g,\g]$ and an arbitrary complementary subspace:
$$
	\zenter(\g) = \left(\zenter(\g) \cap [\g,\g]\right) \oplus V.
$$
Then it is easy to see that $V$ can be separated from $\g$ as a direct summand. \newpage After separating an Abelian summand, we may assume that $\zenter(\g) \subset [\g,\g]$. This means that the center is spanned by $\{h_{i}\}$. Decompose the subalgebra $\h$ as follows.
$$
	\h = \langle h_{1}, \dots h_{n} \rangle \oplus T.
$$
Since $\lambda_{i}(h_{j}) = 0$ for all $i$ and $j$, the linear functions $\lambda_{1}, \dots, \lambda_n$ are linearly independent on $T^{*}$. Moreover, since $T \cap \zenter(\g) = 0$,
 for each $t \in T$ there exists $i$ such that $\lambda_{i}(t) \neq 0$. Therefore, the set of $\lambda_{1}, \dots, \lambda_n$ is a basis in $T^{*}$ and we can choose a basis $t_{1}, \dots ,  t_{n}$ in $T$ such that
$$
	\lambda_{i}(t_{j}) = \delta_{ij}.
$$
Consequently, $\g$ is generated by $e_{i}, e_{-i}, h_{i}, t_{i}$ with the following non-zero relations
\begin{align*}
	[e_{i}, e_{-i}] = h_{i}, \quad
	[t_{i}, e_{i}] = e_{i}, \quad
	[t_{i}, e_{-i}] = -e_{-i}.
	\end{align*}
If $h_{1}, \dots, h_n$ were linearly independent, then $\g$ could be decomposed into a direct sum of subalgebras isomorphic to $\foc$. Since $h_{1}, \dots, h_n$ are not necessarily independent, $\g$ is a quotient of a direct sum by a central ideal, which completes the proof.

\end{proof}
  \subsection{Proof of Theorem \ref{nlpr}}\label{nlpProofReal}

The classification of all real non-degenerate linear pencils can be obtained by describing real forms of the algebras from the list (\ref{ndadecomp_c}). However, we need to know the types of singularities corresponding to these real forms, so it is better to give an explicit classification of real non-degenerate linear pencils.
Proof of the first part of Theorem \ref{nlpr} follows from the following.
\begin{lemma}\label{ndlpLemmaReal}
Let $\g$ be a real Lie algebra, and let $\h \subset \g$ be an Abelian subalgebra. Then the following two conditions are equivalent.
\begin{enumerate}
\item The Lie algebra $\g$ admits a decomposition
\begin{align}\label{ndAlgebraDecompReal}
  	\g \otimes \Complex  = \h \otimes \Complex  + \sum\limits_{i=1}^{n}(V_{\lambda_{i}} + V_{-\lambda_{i}}),
  \end{align}
where  \begin{enumerate} \item $\lambda_1, \dots, \lambda_n \in \h^* \otimes \Complex$ are linearly independent. \item Each space $V_{\pm \lambda_i}$ is one-dimensional. \item For all $\xi \in\h, e_{\lambda_i} \in V_{\lambda_i}$ we have $[\xi, e_{\lambda_i}] = \lambda_i(\xi)e_{\lambda_i}$. \end{enumerate}
\item The Lie algebra $\g$ can be represented as
		\begin{align}\label{ndAlgebraReprReal}
			\g \simeq \bigoplus \so(3) \oplus \bigoplus \sL(2) \oplus \bigoplus \so(3, \Complex) \oplus \left( \bigoplus \ELL  \oplus \bigoplus \hyp    \oplus \bigoplus \foc \right) / \lCal_{0}   \oplus V,
	\end{align}	
	where $V$ is Abelian, $\lCal_{0}$ is a central ideal, and $\h \subset \g$ is a Cartan subalgebra.
\end{enumerate}
\end{lemma}
\begin{proof}
Analogously to the previous section, it suffices to prove the implication 1 $\Rightarrow 2$, namely that $\g$ admits representation (\ref{ndAlgebraReprReal}).\par
Decomposition (\ref{ndAlgebraDecompReal}) implies that $\g$ can be represented as
\begin{align}\label{ndAlgebraDecompReal2}
	\g = \h + \mathrm{span}( e_{\pm 1},\dots, e_{\pm k}, f_{\pm 1}, \dots, f_{\pm l}, \hat e_{\pm 1},  \hat f_{\pm 1}, \dots,  \hat e_{\pm m},  \hat f_{\pm m} ),
\end{align}
with
\begin{align*}
	[h, e_{\pm i}] = \pm \lambda_{i}(h)e_{\pm i}, \quad [h, \hat e_{\pm i}] = 	\pm\xi_{i}(h) \hat e_{\pm i} \mp \eta_{i}(h) \hat f_{\pm i},\\	
	[h, f_{\pm i}] = \pm \mu_{i}(h)f_{\mp i},  \quad  [h, \hat f_{\pm i}] = 	\pm\eta_{i}(h) \hat e_{\pm i} \pm \xi_{i}(h) \hat f_{\pm i} 
\end{align*}
and $\lambda_1, \dots, \lambda_k, \mu_1, \dots, \mu_l, \xi_1, \eta_1, \dots, \xi_m, \eta_m \in \h^*$ being independent.
Analogously to the previous section, the non-zero relation are
\begin{align*}  [e_{i}, e_{-i}] \in \h,
	 \quad [f_{i}, f_{-i}] \in \h, \quad
			[\hat e_{i}, \hat e_{-i}] = -[\hat f_{i}, \hat f_{-i}] \in \h,	\quad  [\hat e_{i}, \hat f_{-i}] = [\hat f_{i}, \hat e_{-i}] \in \h.
\end{align*}
 Suppose that $\lambda_{i}([e_{i}, e_{-i}]) \neq 0$ for some $i$. In this case the triple $e_{i}, e_{-i}, [e_{i}, e_{-i}]$ spans a subalgebra isomorphic to $\sL(2, \R)$. It can be shown that this subalgebra can be separated as a direct summand, analogous to the complex case.\par
Similarly, for $\mu_{i}([f_{i}, f_{-i}] ) \neq 0$, we obtain a summand isomorphic to $\so(3, \R)$ if $\mu_{i}([f_{i}, f_{-i}] ) > 0$, or a summand isomorphic to $\sL(2, \R)$ if $\mu_{i}([f_{i}, f_{-i}] ) < 0$.\par
Further, it is easy to see that
\begin{align*}
	\eta_{i}([\hat e_{i}, \hat e_{-i}]) = \xi_{i}([\hat e_{-i}, \hat f_{i}]) \mbox{\quad and \quad}	 \xi_{i}([\hat e_{i}, \hat e_{-i}]) = -\eta_{i}([\hat e_{-i}, \hat f_{i}]),
\end{align*}
so $\xi_{i}$ and $\eta_{i}$ are either linearly independent on the subspace generated by $ [\hat e_{i}, \hat e_{-i}], [\hat e_{-i}, \hat f_{i}] $, or both vanish. In the first case the elements $\hat e_{i}, \hat e_{-i}, \hat f_{i}, \hat f_{-i},  [\hat e_{i}, \hat e_{-i}], [\hat e_{-i}, \hat f_{i}]$ generate a subalgebra isomorphic to $\so(3, \Complex)$. \par
After separating all simple summands, we may assume that
\begin{align*}
	[e_{i}, e_{-i}] \in \zenter(\g), \quad [f_{i}, f_{-i}] \in  \zenter(\g), \quad
	[\hat e_{i}, \hat e_{-i}] = -[\hat f_{i}, \hat f_{-i}] \in  \zenter(\g),\quad	[\hat e_{i}, \hat f_{-i}] = [\hat f_{i}, \hat e_{-i}] \in  \zenter(\g).	
\end{align*}
In the way absolutely similar to the complex case, it can be shown that $\g$ is decomposed into the sum of an Abelian algebra and the quotient of the sum of several copies of $\hyp, \ELL, \foc$ by a central ideal, which completes the proof.
\end{proof}


\subsection{Proof of the second part of Theorem \ref{nlpr}}\label{typeProof}

By Proposition \ref{indRoots}, the type of $\mathsf{Sing}(\Pi^{\g, A})$ is the triple $(l,k,m)$, where $l,k,m$ are the numbers entering decomposition (\ref{ndAlgebraDecompReal2}).
From the proof of Lemma \ref{ndlpLemmaReal}, we conclude:
\begin{enumerate}
\item Each quadruple $(\pm \xi_{j}, \pm \eta_{j})$ gives an $\so(3, \Complex)$ or $\foc$ summand in (\ref{ndAlgebraReprReal}). Therefore, $m$ coincides with the number of summands isomorphic to $\so(3, \Complex)$ or $\foc$.\par
\item For a pair $\pm \mu_j$, there are three possibilities.
\begin{enumerate}
	\item$\mu_{j}([f_{j}, f_{-j}] ) = 0 \Rightarrow \ELL$.
	\item$\mu_{j}([f_{j}, f_{-j}] ) < 0 \Rightarrow \sL(2, \R)$.
	\item $\mu_{j}([f_{j}, f_{-j}] ) > 0 \Rightarrow \so(3, \R)$.
\end{enumerate}
Consider the second case and calculate the Killing form on $z = [f_{j}, f_{-j}]$. Since
\begin{align*}
	[z, f_{j}] = \mu_{j}(z)f_{-j}	\mbox{\quad and \quad}[z, f_{-j}] = -\mu_{j}(z)f_{j},
\end{align*}
the value of the Killing form on the element $z$ is equal to $\tr (\ad z)^{2} = -2\mu_{j}(z)^{2} < 0$.\par
\item For a pair $\pm \lambda_j$, there are two possibilities.
\begin{enumerate}
	\item $\lambda_{j}([e_{j}, e_{-j}]) = 0 \Rightarrow \hyp$.
	\item $\lambda_{j}([e_{j}, e_{-j}]) \neq 0 \Rightarrow \sL(2, \R)$.
\end{enumerate}
Consider the second case and calculate the Killing form on $z =[e_{j}, e_{-j}]$. Since
\begin{align*}
	[z, e_{j}] = \lambda_{j}(z)e_{j} \mbox{\quad and \quad} 	[z, e_{-j}] = -\lambda_{j}(z)e_{-j},
\end{align*}
\smallskip
the value of the Killing form on $z$ is $\tr (\ad z)^{2} = 2\lambda_{j}(z)^{2} > 0$.
\end{enumerate}
We conclude that the number $l$ is equal to the number of $\ELL$ terms in (\ref{ndAlgebraReprReal}) plus the number of $\so(3, \R)$ terms plus the number of $\sL(2, \R)$ terms such that the Killing form is negative on $\sL(2, \R) \cap \Ker A$.\par Analogously, the number $k$ is equal to the number of $\hyp$ terms plus the number of  $\sL(2, \R)$ terms such that the Killing form is positive on $\sL(2, \R) \cap \Ker A$, q.e.d.

\newpage 
\section*{Appendix A: Jordan-Kronecker decomposition}\label{JKT}

The theorem below describes a simultaneous normal form for an arbitrary pair of skew-symmetric  forms.  Almost all properties of a pencil $\Pi=\{P_\lambda\}$ of compatible Poisson brackets at a fixed point $x$, we used in this paper, can be easily derived from this linear algebraic result. 

\begin{theorem}[Jordan-Kronecker theorem \cite{Gurevich, Gantmacher,Thompson, GZ2, Kozlov}]
\label{JKD}
Let $ A$  and  $ B$ be two skew-symmetric bilinear forms on a complex vector space $V$.
Then there is a basis in $V$ where $A$ and $B$ have the following canonical block-diagonal form:
$$ A  = \begin{pmatrix}
 A_1 &   &    &  \\
  &   A_2  & & \\
  & & \ddots & \\
  & & &   A_k
  \end{pmatrix},
  \qquad
 B = \begin{pmatrix}
 B_1 &   &    &  \\
  &   B_2  & & \\
  & & \ddots & \\
  & & &   B_k
  \end{pmatrix}
  $$
where the pairs of the corresponding blocks $ A_i$ and $ B_i$ can be of the following three types:
\begin{enumerate}
\item Jordan type ($\lambda_i \in \Complex$)
\begin{align*}
A_i = \begin{pmatrix}
     &   \!\!\! J(\lambda_i) \\  & \\
   \!  -J^\top(\lambda_i) &
 \end{pmatrix}, \quad
B_i =  \begin{pmatrix}
   & \ \ -\mathrm{E} \ \  \\  &  \\
   \mathrm{E} &
 \end{pmatrix}
\end{align*}
where $J(\lambda_i)$ denotes the standard Jordan block
$$
J(\lambda_i)=\begin{pmatrix}
\lambda_i&1&&\\
&\lambda_i&\ddots&\\
&&\ddots&1\\&&&\lambda_i\end{pmatrix},
$$
and $\mathrm E$ stands for the identity matrix.
\item Jordan type with $\lambda = \infty$
%
\begin{align*}
A_i =   \begin{pmatrix}
   & \ \ - \mathrm{E}\,\,  \\   &  \\
   \,\mathrm{E} &
 \end{pmatrix},
   \quad
B_i = \begin{pmatrix}
     &   \!\!\!\! J(0) \\    &  \\
   \!  -J^\top(0) &
 \end{pmatrix}.
 \end{align*}
%

\item Kronecker type
\begin{align*} A_i = \begin{pmatrix}
     &   \!\!\!\! S \,\,\\    &  \\
   \!  -S^\top &
 \end{pmatrix},
 \quad B_i = 
  \begin{pmatrix}
     &   \!\!\!\! T \,\, \\    &  \\
   \!  -T^\top &
 \end{pmatrix}
 \end{align*}  

where $S, T$ are 
$$
S = \begin{pmatrix}
  1 & 0  & & \\
      &\ddots & \ddots & \\
     &           &  1 & 0
     \end{pmatrix}, \quad T =  \begin{pmatrix}
  0 & 1  & & \\
      &\ddots & \ddots & \\
     &           &  0 & 1
     \end{pmatrix}.
$$
\end{enumerate}
\end{theorem}

\vskip 20pt
The Jordan-Kronecker form has the following important properties.
\begin{itemize}
\item If $(A_i, B_i)$ is a Jordan pair, then a generic combination $A_i+\lambda B_i$ is non-degenerate, and
$$
\rank (A_i+\lambda B_i) < \max \mbox{ if and only if } \lambda = \lambda_i.
$$
\item If $(A_i, B_i)$ is a Kronecker pair, then any combination $A_i+\lambda B_i$ is degenerate, and
$$
\rank (A_i+\lambda B_i) = \const.
$$
\item  If $(A_i, B_i)$ is a Kronecker pair, then the space
$$
L = \sum\limits_{\lambda} \Ker(A_i+\lambda B_i)
$$
is maximal isotropic with respect to any form $A_i+\lambda B_i$.
\end{itemize}
\bibliographystyle{unsrt}
\bibliography{SBS}
\end{document}